\documentclass[11pt]{article}
\usepackage[T1]{fontenc}
\usepackage[utf8]{inputenc}
\usepackage{lmodern}
\usepackage[margin=1in]{geometry}
\usepackage{microtype,amsmath,amssymb,amsthm,booktabs,array,graphicx,xcolor}
\usepackage[numbers,sort&compress]{natbib}
\usepackage{xurl}
\usepackage[colorlinks=true,allcolors=blue!50!black]{hyperref}
\hypersetup{pdftitle={Revision-Bound Evidence Cuts: Authenticated Queries and Contextual Witness Capacity},pdfauthor={Faruk Alpay and Taylan Alpay}}
\usepackage{listings,caption}
\newtheorem{theorem}{Theorem}[section]
\newtheorem{proposition}[theorem]{Proposition}
\newtheorem{lemma}[theorem]{Lemma}
\newtheorem{definition}[theorem]{Definition}
\newtheorem{corollary}[theorem]{Corollary}
\newcommand{\N}{\mathbb{N}}

\newcommand{\W}{\mathcal{W}}
\newcommand{\Enc}{\operatorname{Enc}}
\newcommand{\cl}{\operatorname{cl}}
\newcommand{\Erase}{\operatorname{Erase}}
\newcommand{\supp}{\operatorname{supp}}
\title{Revision-Bound Evidence Cuts:\\Authenticated Queries and\\Contextual Witness Capacity}
\author{\begin{minipage}[t]{0.46\textwidth}\centering
Faruk Alpay\thanks{Correspondence: \href{mailto:alpay@lightcap.ai}{alpay@lightcap.ai}.}\\[3pt]
\small Department of Computer Engineering\\
Bah\c{c}e\c{s}ehir University\\Istanbul, T\"urkiye\\
\footnotesize\href{mailto:faruk.alpay@bahcesehir.edu.tr}{\texttt{faruk.alpay@bahcesehir.edu.tr}}
\end{minipage}\hspace{0.03\textwidth}
\begin{minipage}[t]{0.46\textwidth}\centering
Taylan Alpay\\[3pt]
\small Department of Aerospace\\
University of Turkish Aeronautical Association\\Ankara, T\"urkiye\\
\footnotesize\href{mailto:s220112602@stu.thk.edu.tr}{\texttt{s220112602@stu.thk.edu.tr}}
\end{minipage}}
\date{27 October 2026}
\begin{document}
\maketitle
\begin{abstract}
Authentic sources can be recombined into claims they do not justify, and an old
inference can outlive its qualifications. We construct revision-bound evidence
cuts that retain supporting and opposing obligations under an authenticated
update frontier. Disjoint cuts transport across unrelated changes; alternative
cuts and an exact transversal-based reuse radius expose the cost of revision.
Their source footprints enter contextual witness capacity with conjunctive
obligations and consumable authorizations.
A relational query contract binds counting units, parameters, complete selected
bags and aggregates across SQLite scans, indexes and materialized views. The
same full-source verifier checks all 16,824 Chogha Mish records; three query
strategies pass and eight contract/view mutations fail. The source documents
repeated jaw/tooth records and lacks demonstrated coeval deposits, so its
postcranial contrast remains a qualified hypothesis. An independent English
comparison gives a negative partial transfer; exact floor/refuse replication
remains unidentified. A separate 340-row NISP comparison reverses direction
under two admitted weighting rules.
A frozen nonlearned English compiler maps 243 real documentary records to
qualified span frames. On 55 source-group-held-out records, uncertainty-cue
F1 is 0.75; sparse positive support and five conservative revision failures are
reported. External audits on 2,965 human temporal annotations expose substantial
abstention and limited recall. Fresh text-to-frame receipts remain mandatory
when unchanged semantic cuts are reused.
Source-exact reconstruction separates logical-source and transport commitments,
restores original OBJ/JPEG bytes, and enforces both ZIP and opened-file budgets.
The artifact includes code, core sources, provenance, fixed rules, errors and
adversarial experiments; external corpus acquisition is separately documented.
\end{abstract}
\section{Problem: authentic evidence in an unauthorized scene}
An authenticated archaeological measurement can support a geometric approximation.
It cannot by itself establish the truth of a caption, the intentions of a past
society, or the source institution's approval of a contemporary composition.
Likewise, displaying several recognizable authorities together does not establish
that their authority is interchangeable or jointly available without restriction.
The engineering problem is to preserve the missing relations explicitly.

We distinguish three predicates: $\mathrm{Attributed}(e,s)$ records the claimed
origin of an evidence item. The predicate $\mathrm{Supports}(e,q,C)$ records
a qualified support relation under a declared context.
Finally, $\mathrm{Endorses}(s,S,C,g)$ records a signed,
policy-limited authorization for a particular scene. Neither of the first two
predicates entails the third. The signature authenticates an assertion and its
issuer, conditional on trusted keys; it does not make the assertion historically true.

The motivating media hypothesis is \emph{contextual authority transfer}: viewers
may interpret the presence of an authentic source or familiar symbol as approval
of a newly assembled narrative. Bourdieu supplies a conceptual account of
recognition and symbolic power~\cite{bourdieu1989}; Hall's encoding/decoding model
allows negotiated and divergent reception~\cite{hall1973}. Framing~\cite{entman1993},
source credibility~\cite{hovland1951}, source monitoring~\cite{johnson1993}, and
repetition effects~\cite{fazio2019} motivate variables for future experiments.
These sources do not establish that symbolic power is literally conserved,
that media producers share a coordinated intention, or that repeated reuse must
end in collapse. We replace those unsupported conclusions with a precise,
falsifiable computational question: \emph{which authenticated obligations disappear
when a presentation is reduced to a list of source names?}

The contribution begins with the missing distinction between support,
contradiction, uncertainty and authorization. Revision-bound polarity cuts retain
both decision polarities and mandatory qualifications; a verified change frontier
determines whether an old cut can be transported or replaced by an avoiding cut.
A derived reuse threshold then connects the certificate family to a transversal
problem. Its source footprints enter contextual witness capacity, which counts
complete presentations under conjunctive obligations and consumable grants.
The SCENE artifact implements the component mechanisms over original source bytes,
recorded transformations, typed interpretations and explicit local issuers.
The method is investigator-designed and inspectable, with no learned component.

The relational contribution is a reconstruction-aware authenticated query
contract. It distinguishes complete selected bags from grouped answers, binds
counting units and original-column mappings, and verifies equivalent scan,
index and materialized-view executions against the same full-source scope
verifier. Thus a genuine aggregate cannot answer a different parameter, unit
or source-version contract merely because the displayed ratio agrees. The
actual SQLite workload retains all 59 source fields and 16,824 records; source
verification, query execution and materialization costs are measured separately.

A language-layer experiment compiles actual English source descriptions and
catalogue-date phrases into explicitly qualified span frames. Frozen rules,
source-group holdout, exact component scores and controlled text revisions
measure extraction separately from receipt validity. External human-annotated
historical corpora test transfer without changing those rules. The resulting
receipts bind text, extraction policy and semantic slots; successful verification
does not make a misparsed qualification linguistically correct.

\paragraph{Relation to earlier work.}
Our previous studies distinguish provenance from authorization in autobiographical
models~\cite{alpay2026provenance} and bind approval to published content and recovery
state~\cite{alpay2026approval}. The present construction adds authenticated
revision transport for polarity-complete inference certificates, a cut-frontier
reuse threshold, and a link from surviving certificate alternatives to integral
endorsement capacity. We derive these objects from established logic, provenance,
authentication and combinatorial optimization; their exact scope and nearest
technical precedents are given with the construction below.

\section{Mathematical lineage and the point of failure}
\label{sec:history}
The relevant history comprises several questions that should not be collapsed
into one mathematical tool. Hall's 1935 representative theorem~\cite{hall1935}
addresses choosing one permissible representative per demand. Gale, Kuhn and
Tucker's 1951 treatment of linear-programming duality~\cite{gale1951} supplies the
price certificates used below. Ford and Fulkerson's 1956 maximum-flow/minimum-cut
result~\cite{ford1956} turns suitable capacity constraints into integral allocations.
Karp's 1972 reductions~\cite{karp1972} locate the obstruction when joint choices
become combinatorial packing.

Ralph Merkle's authenticated-tree construction, represented here by his 1980
paper~\cite{merkle1980}, addresses efficient integrity authentication. It does not
solve witness admissibility or scarce-resource allocation. Girard's 1987 linear
logic~\cite{girard1987} distinguishes reusable truth from consumable premises;
Bowers et al. implemented consumable credentials in authorization systems in
2007~\cite{bowers2007}. Green, Karvounarakis and Tannen's 2007 provenance
semirings~\cite{green2007} preserve derivation structure and multiplicity. The
failure below is therefore a \emph{lossy projection of provenance}, not a limitation
of every provenance algebra.

A source set records which names occur but loses occurrence counts and the
distinction between prerequisites and alternatives. A provenance polynomial
retains derivations and multiplicity; an admissibility policy must still specify
which witnesses are permitted and which signed rights remain available.
Hall conditions and network cuts solve allocation among singleton substitutes.
Joint prerequisites instead require a conjunctive allocation model, whose
fractional price bounds do not guarantee integral feasibility.

Authentication and resource state address different remaining questions.
A committed root binds records against alteration; truth and admissibility
depend on explicitly admitted evidence and rules. Serialized debits account
for spent capacity, while freshness requires an accepted current state.
Neither layer by itself determines audience interpretation. These distinct
information requirements motivate the following occurrence-preserving construction.

\subsection{The idempotence obstruction}
Let $A$ be the support set of one displayed object. The support set of two copies
is still $A\cup A=A$. Suppose a positive additive exposure tariff factors through
source sets as $f(A\cup B)=f(A)+f(B)$, taking values in a cancellative commutative
monoid. Then
\begin{equation}
 f(A)=f(A\cup A)=f(A)+f(A)\quad\Longrightarrow\quad f(A)=0.
 \label{eq:idempotence}
\end{equation}
Thus a nonzero additive charge cannot be recovered from an idempotent union
summary. This elementary argument identifies the required boundary: deduplicate
a source inside \emph{one object's evidence closure}, but retain multiplicity
across distinct rendered occurrences. A second loss is connective structure:
the union $\{a,b\}$ cannot distinguish ``both $a$ and $b$'' from ``either $a$ or $b$.''
An untyped source list loses both pieces of information before any cryptography
is applied.

\section{Evidence closures and context-admissible witnesses}
\label{sec:model}
Let $\mathcal{S}$ and $\mathcal{O}$ be finite nonempty sets of source principals
and object types, and let $\mathcal{E}$ be a finite directed acyclic evidence
graph. Each object has a nonempty set of root evidence records. Each record names
a source, a statement, a claim class, qualifications and its complete list of
mandatory supporting records. For root records $R_o$, let $\cl(R_o)$ be the
reflexive transitive support closure. In a fixed reconstruction,
\begin{equation}
 w_o=\{s(e):e\in\cl(R_o)\},\qquad
 c_s(S)=\sum_{o\in\mathcal{O}}m_o\,\mathbf{1}[s\in w_o],
 \label{eq:closure-cost}
\end{equation}
where $m_o\in\mathbb{Z}_{\geq1}$ counts rendered occurrences of type $o$. A source contributes
once to an occurrence even if several support paths reach it; a second displayed
occurrence contributes another unit. A unit is a protocol-defined endorsement
authorization, not a viewer, impression, copyright permission or measure of prestige.

\begin{definition}[Admissible witness family]
A context $C$ specifies a purpose, audience and explicit evidence-admissibility
policy. For each type $o$, the finite nonempty family
$\W_C(o)\subseteq2^{\mathcal{S}}\setminus\{\varnothing\}$ contains acceptable witness
bundles. Every source in a selected bundle is required (AND); selecting a member
of the family is a choice among acceptable realizations (OR).
\end{definition}

For a fixed manifest, $\W_C(o)=\{w_o\}$. Alternatives require separately declared
and reviewed realizations. An optimizer does not infer that two historical
interpretations are equivalent. Budget $Q_s\in\N$ is the remaining capacity of a
selected grant for source $s$; additional grant issuance is outside that bound.
Every possible rendition must retain its chosen witness and obtain fresh
authorization for its exact resulting root. Here a rendition means one new
authorization event; rereading an already accepted artifact is not another
rendition or debit. Actual display frequency is outside this accounting model.

\paragraph{Compatibility assumption.}
Witness choices for different object occurrences are independently compatible
within a rendition. This excludes, unless explicitly encoded, constraints such
as ``all objects must follow the same reconstruction school.'' Such applications
need complete-scene configuration variables or additional coupling constraints.
This assumption is necessary for the counting equivalence in the next section.

\section{Deriving contextual witness capacity}
\label{sec:capacity}
\begin{definition}[Contextual witness capacity]
$K_C(Q,m)$ is the largest integer number of complete scene renditions whose
occurrences can be assigned admissible witnesses without exceeding $Q$.
\end{definition}
Put $a_{s,o,w}=\mathbf{1}[s\in w]$. Counting how many occurrences use each witness
gives the configuration integer program
\begin{align}
 K_C(Q,m)=\max\;&k \label{eq:ilp}\\[-3pt]
 \text{subject to }&\sum_{w\in\W_C(o)}x_{o,w}=k m_o&&\text{for all }o,\nonumber\\
 &\sum_{o,w}a_{s,o,w}x_{o,w}\leq Q_s&&\text{for all }s,\nonumber\\
 &x_{o,w},k\in\N.\nonumber
\end{align}
We include zero in $\N$. An actual collection of renditions yields these counts.
Conversely, integral counts provide $k m_o$ occurrences of each type; partitioning
each type into $k$ groups of size $m_o$ gives complete renditions under the
compatibility assumption. Hence the formulation is exact, rather than merely
a necessary aggregate budget condition.

\subsection{The provenance-polynomial form}
For formal variables $z_s$, define
\begin{equation}
 P_{o,C}(z)=\sum_{w\in\W_C(o)}\prod_{s\in w}z_s,
 \qquad P_C(z)=\prod_o P_{o,C}(z)^{m_o}.
 \label{eq:polynomial}
\end{equation}
A bundle is a squarefree monomial: mandatory sources multiply. Alternative
bundles add. Raising to the number of occurrences retains repeated exposure in
exponents. Polynomial provenance is established machinery~\cite{green2007}; its
application here yields the equivalent capacity expression
\begin{equation}
 K_C(Q,m)=\max\{k\in\N:\exists d\in\supp(P_C^k),\ d_s\leq Q_s\ \forall s\}.
 \label{eq:support-capacity}
\end{equation}
Here $\supp(P)=\{d\in\N^{\mathcal S}:[z^d]P>0\}$ is the set of exponent vectors
of nonzero monomials, and $P_C^0=1$, so $k=0$ is always feasible.
The coefficients count derivations producing an exponent vector; feasibility
depends on exponent support. For one fixed witness per type, the entire problem
reduces to
\begin{equation}
 K_{\rm fixed}(Q)=\min_{s:c_s(S)>0}\left\lfloor\frac{Q_s}{c_s(S)}\right\rfloor.
 \label{eq:fixed}
\end{equation}
This is a capacity of explicitly limited rights, not a predicted cultural life span.

\subsection{Where matching remains exact}
Assume every witness is a singleton. For nonempty $T\subseteq\mathcal{O}$ let
$N_C(T)=\bigcup_{o\in T}\{s:\{s\}\in\W_C(o)\}$ and $m(T)=\sum_{o\in T}m_o$.
The capacitated form of Hall's condition gives
\begin{equation}
 K_C(Q,m)=\min_{\varnothing\ne T\subseteq\mathcal{O}}
 \left\lfloor\frac{\sum_{s\in N_C(T)}Q_s}{m(T)}\right\rfloor.
 \label{eq:hall}
\end{equation}
Necessity follows because all $km(T)$ occurrences must use $N_C(T)$. For
sufficiency, a network sends capacity $km_o$ from a super-source to each type,
then to its admissible source nodes, then through capacities $Q_s$ to the sink.
The displayed conditions ensure every cut can cover total demand; integral
maximum flow supplies witness counts~\cite{hall1935,ford1956}. This is a
specialization of known matching theory, not a new Hall theorem.

\subsection{Deriving the source-price equation}
Relax $x$ and $k$ to nonnegative reals and denote the optimum by $F_C(Q,m)$.
For a nonnegative price vector $y$, define
\begin{equation}
 p_{o,C}(y)=\min_{w\in\W_C(o)}\sum_{s\in w}y_s,
 \qquad D_C(y)=\sum_o m_o p_{o,C}(y).
 \label{eq:price}
\end{equation}
Every feasible allocation obeys
\begin{equation}
 Q\cdot y\ \geq\ \sum_{o,w}x_{o,w}\sum_{s\in w}y_s
 \ \geq\ k\sum_o m_o p_{o,C}(y)=kD_C(y).
 \label{eq:bound}
\end{equation}
The inequality produces a checkable bottleneck certificate. It becomes an exact
fractional equation by duality:
\begin{theorem}[Contextual source-price representation]
For the finite nonempty witness families above,
\begin{equation}
 F_C(Q,m)=\inf_{\substack{y\geq0\\D_C(y)>0}}
 \frac{Q\cdot y}{D_C(y)},\qquad
 K_C(Q,m)\leq\lfloor F_C(Q,m)\rfloor.
 \label{eq:dual}
\end{equation}
\end{theorem}
\begin{proof}
Write the relaxed primal as $\max k$ subject to $Ax\leq Q$, $Ex=mk$,
$x,k\geq0$, where $E$ identifies each variable's object type. Its dual is
\begin{equation}
 \min_{y\geq0,z} Q\cdot y:\quad
 z_o\leq\sum_{s\in w}y_s\ (o,w),\quad \sum_o m_o z_o\geq1.
 \label{eq:dual-lp}
\end{equation}
For fixed $y$, object potentials may be chosen as their largest admissible
values $z_o=p_{o,C}(y)\geq0$. Feasibility becomes $D_C(y)\geq1$.
Every occurrence uses at least one source, so the zero-feasible primal is bounded
by finite total capacity. Strong linear-programming duality~\cite{gale1951}
therefore applies. Homogeneity normalizes $D_C(y)$ to one and gives the ratio.
Equation~\eqref{eq:bound} and integrality give the last inequality.
\end{proof}
Zero source budgets are permitted, including $Q=0$, for which both capacities
are zero. The normalization never divides by a budget: a strictly positive
price on every source gives $D_C(y)>0$ because objects and witness bundles are
nonempty.
The source prices identify scarcity in the \emph{declared policy}. They do not
score historical truth. A rational $y$ provides an exact upper bound even if a
numerical optimizer found it; optimality additionally needs a matching primal
certificate or an independent argument.

\subsection{The AND bottleneck: complexity and integrality}
The singleton flow formula cannot simply be reused for joint prerequisites.
\begin{proposition}[Width-two feasibility obstruction]
Given explicitly listed witness families, deciding $K_C\geq1$ is NP-complete
even with unit source budgets, unit object multiplicities and witnesses of
size exactly two.
\end{proposition}
\begin{proof}
Take a perfect three-dimensional matching instance with nonempty equal-sized,
pairwise disjoint sets $X,Y,Z$. If some $x$ belongs to no triple, map the instance
to the fixed valid NO instance consisting of two object types, both requiring
the sole witness $\{a,b\}$ with $Q_a=Q_b=1$. Otherwise, for each $x\in X$ create
one object type; a triple $(x,y,z)$ supplies its
admissible witness $\{y,z\}$. Give every source in the disjoint union $Y\sqcup Z$
capacity one. Choosing one nonconflicting witness for every $x$ is exactly a
perfect three-dimensional matching. This is a polynomial reduction from the
classical problem~\cite{karp1972}. A selected witness per object is a polynomial
certificate, establishing membership in NP. Empty three-dimensional matching
inputs can separately be mapped to one object requiring $\{a,b\}$ with unit
budgets, a fixed YES instance. Thus every output respects the model's nonempty
witness-family requirement.
\end{proof}
This classification is a consequence of known combinatorial hardness. It is
not an average-case cryptographic hardness assumption, and we do not use it to
derive encryption keys or a trapdoor.

For a concrete integrality failure, use seven unit-budget sources and one type
whose witness triples are
\begin{equation}
 012,\ 034,\ 056,\ 135,\ 146,\ 236,\ 245.
 \label{eq:fano}
\end{equation}
Each pair of triples intersects; each source occurs three times. Thus $K_C=1$.
Assigning $x_w=1/3$ to each triple is fractionally feasible. Summing source
constraints gives $3\sum_w x_w\leq7$, so $F_C=7/3$. Consequently
$K_C=1<\lfloor F_C\rfloor=2$. This explicitly constructed Fano-plane incidence
example separates integral authorization from its fractional certificate. A
verifier that accepts a fractional plan as a complete rendition would be unsound.

\subsection{A derived measure: obligation-erasure amplification}
\label{sec:erasure}
Consider the particular defective operation
\begin{equation}
 \Erase(\W_C(o))=\{\{s\}:s\in\bigcup_{w\in\W_C(o)}w\},\qquad
 \mathcal{A}_C(Q,m)=\frac{K_{\Erase(C)}(Q,m)}{K_C(Q,m)}\quad(K_C>0).
 \label{eq:amplification}
\end{equation}
It replaces every conjunction by singleton alternatives. The ratio measures
excess capacity caused by this specific loss of obligation structure. At zero
true capacity we report infeasibility explicitly rather than divide by zero.
The notation $\Erase(C)$ means the same instance with
$\W_{\Erase(C)}(o)=\Erase(\W_C(o))$ and all other inputs unchanged.
Any true feasible plan remains feasible after replacing each bundle by one of
its members, so $K_{\Erase(C)}\geq K_C$.

\begin{theorem}[Unbounded erasure amplification at width two]
For every integer $n\geq1$, there is a unit-budget instance with maximum witness
size two, $K_C=1$ and $\mathcal{A}_C=n+1$.
\end{theorem}
\begin{proof}
Use sources $s_0,\ldots,s_n$, each with budget one, and one object type of
multiplicity one. Let
\begin{equation}
 \W_C(o)=\{\{s_0,s_j\}:1\leq j\leq n\}.
 \label{eq:star}
\end{equation}
Every true rendition consumes the shared prerequisite $s_0$, so at most one is
possible and any pair realizes one. After erasure each of the $n+1$ singleton
sources is an alternative and all can be used once. Thus $K_{\Erase(C)}=n+1$.
\end{proof}
There is no bound on amplification depending only on witness width. The lost
shared prerequisite, rather than a broken hash, causes the apparent increase.
In~\eqref{eq:fano}, the three capacities are $1$ (integral), $7/3$ (fractional)
and $7$ (erased); they answer different questions.

\begin{corollary}[Source-list insufficiency]
A summary retaining only source unions, multiplicities and budgets cannot give
a sound constant-factor capacity approximation for all witness policies: no
fixed finite $b\geq1$ and estimate based only on that summary can satisfy
$K_C/b\leq\widehat K_C\leq K_C$ for every instance.
\end{corollary}
\begin{proof}
The star policy and its erased policy have identical such summaries and
capacities $1$ and $n+1$. A common estimate sound for the former is at most one;
an approximation within factor $b$ for the latter must be at least $(n+1)/b$.
Choose $n+1>b$.
\end{proof}

\paragraph{Context restrictions.}
If $\W_{C'}(o)\subseteq\W_C(o)$ with unchanged budgets and multiplicities,
then $K_{C'}\leq K_C$ and $F_{C'}\leq F_C$ by feasible-set inclusion. Merely
changing a context string does not change these quantities unless admissibility
changes. The cryptographic layer separately prevents transferring a permit to
that new string. No general monotonicity of the ratio $\mathcal A_C$ follows.

\begin{figure}[t]
\centering\includegraphics[width=\linewidth]{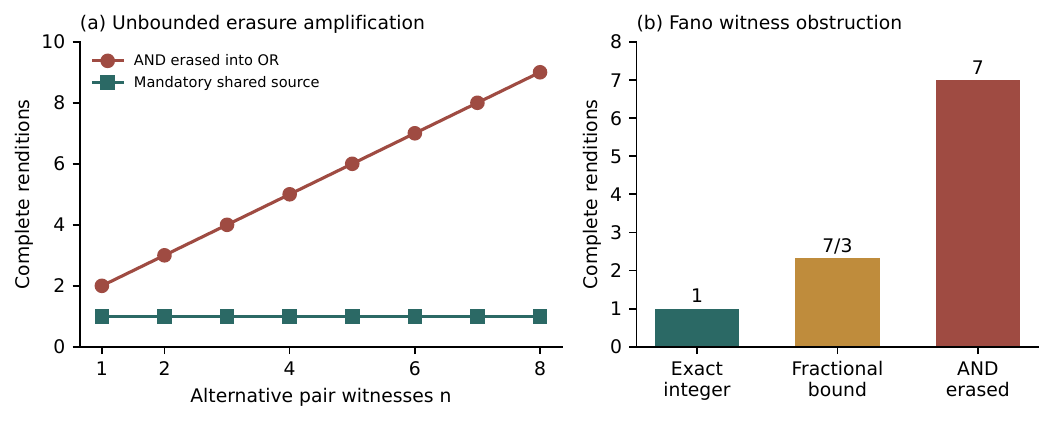}
\caption{Constructed mathematical experiments. (a) Keeping a shared prerequisite
fixes capacity at one; replacing AND by OR gives $n+1$. (b) The Fano fixture
distinguishes an integral plan, a fractional certificate and erased obligations.
These are not measurements of cultural influence.}
\label{fig:capacity}
\end{figure}

\section{Revision-bound polarity cuts}
\label{sec:revisions}
The witness families in Section~\ref{sec:model} specify what a rendition must
retain. We now derive such families from executable inference obligations and
authenticated evidence revisions. The problem is finer than preserving a source
hash: a qualification or registered counterexample can change a presentation's
eligibility even when every underlying source file remains unchanged.

\subsection{Typed evidence states and closed-catalogue guards}
Fix a finite nonempty leaf registry $L$, a scope $C$ and an acyclic policy DAG
$\Pi$. The scope identifies the intended question, purpose and audience. Each
leaf $\ell$ has a complete canonical record
\[
 r_t(\ell)=(\mathrm{id},\mathrm{source},\mathrm{record\_hash},
                    \mathrm{qualifier},p_t,n_t),\qquad p_t,n_t\in\{0,1\}.
\]
The source hash identifies the underlying scientific record; the qualifier and
two bits record its admission for this particular inference. The four states
are unknown $(0,0)$, supported-only $(1,0)$, opposed-only $(0,1)$ and contested
$(1,1)$. They describe admitted evidence, without assigning probabilities or
asserting objective truth. Four-valued evidence semantics are established
\cite{belnap1977}. In particular, an unknown premise does not become supported
by negation, and inconsistency does not entail an unrelated proposition.

Each gate authenticates its ID, operation, complete ordered argument list,
qualifier and type, either \emph{interpretation} or \emph{eligibility}. There is
no authorization gate or inference rule producing issuer consent. For
interpretation gates we use the usual truth-order operations
\begin{align}
 \bigwedge_i(p_i,n_i)&=(\bigwedge_i p_i,\ \bigvee_i n_i),&
 \bigvee_i(p_i,n_i)&=(\bigvee_i p_i,\ \bigwedge_i n_i),&
 \neg(p,n)&=(n,p). \label{eq:polarity-semantics}
\end{align}
Logical disjunction must be distinguished from a choice among individually
admissible witnesses. For example, $(1,1)\vee(0,0)=(1,0)$, although neither input
is supported-only. If an application requires each chosen alternative to be
individually uncontested, that requirement belongs in its reviewed policy; it
does not follow from the FDE disjunction in~\eqref{eq:polarity-semantics}.

An eligibility gate can additionally use an explicit guard
$G(q;d_1,\ldots,d_k)$, $k\geq1$, whose registered defeaters are part of $\Pi$:
\begin{equation}
 p_G=p_q\wedge\bigwedge_{j=1}^k\neg p_{d_j},\qquad
 n_G=n_q\vee\bigvee_{j=1}^k p_{d_j}.
 \label{eq:eligibility-guard}
\end{equation}
Here $p_{d_j}=0$ means no admitted support for this defeater in the committed
catalogue. It is not a proof that the defeater is false in the world. A contested
defeater still has support and blocks eligibility. The guard is a declared
nonmonotone admission rule; it is not identified with FDE negation. An
eligibility decision passes only at $(p_G,n_G)=(1,0)$. The completeness of the
catalogue, scientific suitability of its rules and reliability of admission
remain explicit assumptions.

\subsection{Obligation-preserving local cuts}
For either bit of a gate, expand its defining equation into an AND or OR of
signed child-bit literals; NOT changes the channel, while the guard introduces
the literals $\neg p_{d_j}$. A local certificate applies the following rules:
\begin{equation}
\begin{array}{c|cc}
 &\text{output }1&\text{output }0\\ \hline
 \mathrm{AND}&\text{all literals equal }1&\text{one literal equals }0\\
 \mathrm{OR}&\text{one literal equals }1&\text{all literals equal }0.
\end{array}
\label{eq:cut-rules}
\end{equation}
Both target bits must be proved. A \emph{decision cut} $c$ contains the resulting
reachable proof DAG, complete records for its opened leaves and gates, and
their authenticated openings. Its frontier $F(c)\subseteq L$ is the set of
opened full leaf records. The verifier rejects duplicate or extraneous steps,
cycles, omitted required children, altered argument lists, missing qualifiers
and unsupported status/type changes. A leaf opened for one bit nevertheless
discloses its whole record. This is selective disclosure of unused records,
not zero knowledge about the disclosed frontier.

\begin{lemma}[Local cut soundness]
Fix $\Pi$ and the records on $F(c)$. If $c$ is valid for target $q$ with status
$v$, then every assignment of records to the other leaves, using the same
registry and well-formed bit domain, evaluates $q$ to $v$ under $\Pi$.
\end{lemma}
\begin{proof}
An opened leaf fixes each asserted bit. Induct through the reachable proof
steps in topological order. All-one premises determine a true conjunction; a
single zero determines a false one. Dually, one true premise determines a true
disjunction, while all-zero premises determine a false one. Literal polarity
and channel changes are fixed by the authenticated gate. Applying this
induction to both root channels proves $v$, independently of unopened records.
\end{proof}

This is soundness of a particular local calculus, not completeness for every
semantic certificate of a reconvergent circuit. Algebraic identities may admit
smaller certificates in another proof system. Our avoiding-cut procedure is
complete for the rules~\eqref{eq:cut-rules}: at a forbidden leaf it returns
unavailable; an all-required rule requires all child certificates; a decisive
rule requires one available child with the decisive value. Memoization over
node/channel pairs decides feasibility and selects the first feasible branch.
It does not globally minimize shared source footprint, bytes or cost.
These distinctions are consistent with the established tradeoffs in knowledge
compilation and certificate complexity~\cite{darwiche2002,buhrman2002,nisan1991}.

\subsection{Signed revisions and complete update cuts}
The reference implementation commits the ordered leaf vector using a cached
perfect binary tree padded to the next power of two. A record leaf hashes as
$H(\mathtt{0x00}\Vert\Enc(r))$, an internal node as
$H(\mathtt{0x01}\Vert h_L\Vert h_R)$, and padding uses a separate domain and
its eight-byte position. This fixed perfect-tree encoding differs from the
RFC~9162 tree in the later scene commitments. Gate records have a separately
committed tree. The policy descriptor commits the registry digest, leaf and
gate counts, gate-tree root and scope hash. A checkpoint signs
\begin{equation}
 B_t=(\mathrm{version},h_t,|L|,h_L,h_\Pi,h_C,t,e_t,\nu_t,R_{t-1}),\qquad
 R_t=H(D_R\Vert\mathtt{0x00}\Vert\Enc(B_t)),
 \label{eq:revision-checkpoint}
\end{equation}
with $R_t$ included in the signed body. Here $B_t$ denotes the canonical
named-field object, not an alternative array serialization,
$D_R=\texttt{SCENE/revision-root/1}$, $h_t$ is the leaf-tree root,
$h_L$ the registry digest, $h_\Pi$ the policy digest, $e_t$ an epoch and $\nu_t$
a nonce. The Ed25519 signature covers the canonical body after the distinct
\texttt{SCENE/revision-checkpoint-signature/1} domain and a NUL separator.
The consumer pins the public key,
current $R_t$, $h_\Pi$ and $h_C$ outside the packet, compares the packet target
with its requested $q$, and requires the intended decision type. A correctly
signed alternative checkpoint does not replace these pins. A trusted current
root or external history mechanism supplies freshness and fork resolution;
signatures and increasing sequence numbers alone supply neither.

For consecutive admitted snapshots, define the full-record change set
\begin{equation}
 \Delta_{t,t+1}=\{\ell\in L:r_t(\ell)\ne r_{t+1}(\ell)\}.
 \label{eq:changed-records}
\end{equation}
A \emph{complete update cut} partitions the entire padded tree into disjoint
shared subtrees and opened changed leaves. A shared subtree supplies one hash
used in both reconstructions. A changed leaf supplies its position and complete
old/new records, with equal IDs and unequal records. The verifier reconstructs
both complete roots, checks that every position is covered exactly once and
that no supplied node is unused, and verifies both checkpoints. It requires
the same registry, policy and scope; sequence $t+1$; parent $R_t$;
nondecreasing epoch; and a different nonce. Padding cannot change. A change to
leaf identity, policy or scope requires a new policy outside this interface.

\begin{lemma}[Complete-delta soundness]
Under authentic admitted checkpoints and collision-resistant hashing, an
accepted update cut returns exactly $\Delta_{t,t+1}$.
\end{lemma}
\begin{proof}
Coverage and old/new root reconstruction authenticate the two complete vectors.
Every explicit changed record is unequal. If an actual unequal record were
omitted, it would lie in a claimed shared subtree with equal old/new hash.
The two distinct canonical subtree contents would then yield a hash collision.
The same binding excludes inventing an unequal record outside the admitted
vectors. Coverage rejects overlap and unexplained gaps.
\end{proof}

The initial producer stores $O(|L|)$ leaf-tree hashes. A persistent revision
rehashes the union of affected paths, bounded by $O(k\log |L|)$ for $k$ changed
records, and stores sparse overlays. Cached gate aggregates and reverse edges
update affected gates without repeatedly scanning every argument of a large
fanin gate. Enumerating the update cut skips equal cached subtrees. Initial
ingestion, affected-gate work, proof output and lookup overhead across long
persistent histories remain real costs; the path bound is not a bound on all
producer work.

\subsection{Certificate transport and its security game}
\begin{theorem}[Cut-stable revision transport]
Let $c$ be a valid cut at $R_t$ and let a verified complete update cut connect
admitted checkpoints $R_t,R_{t+1}$ with fixed registry, policy and scope. If
\begin{equation}
 F(c)\cap\Delta_{t,t+1}=\varnothing, \label{eq:transport-disjoint}
\end{equation}
the same local inference steps certify the same target pair at $R_{t+1}$.
Its leaf authentication paths can be repaired using the update cut, yielding
an ordinary independently verifiable decision cut at the new root.
\end{theorem}
\begin{proof}
Complete-delta soundness and~\eqref{eq:transport-disjoint} imply equality of
every opened full leaf record. Gate records are unchanged because policy is
fixed, so local cut soundness preserves both target bits. For each opening,
a sibling subtree exposed in the update reconstruction supplies its new hash;
a sibling inside a shared subtree retains its old hash. Replacing the former
entries reconstructs the new leaf-tree root and checkpoint wrapper. The
ordinary new-root cut verifier checks the repaired result.
\end{proof}

Disjointness is necessary for reusing the \emph{same full-record frontier} by
this rule, not for semantic output stability. A qualifier-only revision belongs
to $\Delta$ even if both bits and the source-file hash are unchanged. A touched
certificate may also have an unaffected alternative: generating an old cut
with $F(c)\cap\Delta=\varnothing$ and then transporting it recovers that case.
If no such cut exists, current reevaluation can still find a new certificate.
The batch API verifies one complete delta once, then validates and repairs
each old cut; its verified map is local to that invocation. This amortizes
delta checking without making individual proof checking free.

For precision, consider an adversary that receives admitted canonical states,
their signed checkpoints and legitimate openings, and outputs a decision cut,
complete delta or transported cut under externally pinned roots, policy, scope
and query. It wins if acceptance certifies a different evaluated pair, an
incorrect changed-record set, or an unchanged frontier containing an unequal
record. The admission authority is honest about the fixed registry/DAG schema;
source truth, policy adequacy and a compromised signing key are outside this
game. With strict schema checking and injective canonical descriptors, there
are polynomial-overhead reductions satisfying
\begin{equation}
 \Pr[\mathrm{win}]\ \leq\
 \operatorname{Adv}^{\rm muEUF\mbox{-}CMA}_{\rm Ed25519}(\mathcal B_{\rm sig})
 +\operatorname{Adv}^{\rm CR}_{\rm SHA256}(\mathcal B_H).
 \label{eq:revision-security}
\end{equation}
Here the multi-user term accommodates the admitted key set; for one checkpoint
key the ordinary EUF-CMA game suffices. An unissued accepted checkpoint gives a
signature forgery. Otherwise, a different committed record, gate, descriptor
or hidden change gives a collision by the tree-binding argument. With neither
event, the two preceding lemmas and transport induction rule out winning.
This is a conditional authenticity reduction for a construction from standard
primitives, not a new hardness assumption or a bound estimated by rejection
counts. Public benchmark seeds instantiate functional tests only.

\subsection{An exact certified-reuse threshold}
Fix an admitted old snapshot and a nonempty family $\mathcal C_t(q)$ of all
valid local cuts for its exact target pair. Its frontier hypergraph is
$\mathcal H_t(q)=\{F(c):c\in\mathcal C_t(q)\}$. Our syntax has no constants or
empty-argument gates, so all frontiers are nonempty. Define
\begin{equation}
 \tau_t(q)=\min\{|D|:D\subseteq L,\ D\cap F\ne\varnothing
                              \text{ for every }F\in\mathcal H_t(q)\}.
 \label{eq:revision-transversal}
\end{equation}
This is the classical transversal number applied to the concrete old-certificate
family. Boolean certificate complexity supplies the underlying mathematical
language~\cite{nisan1991,buhrman2002}; our functional measures authenticated
old-cut reuse rather than the number of queries needed to evaluate a function.

\begin{theorem}[Certified old-cut reuse radius]
The largest integer $r$ for which every changed-record set of size at most $r$
leaves some transportable old cut is exactly
\begin{equation}
 \rho_t(q)=\tau_t(q)-1. \label{eq:transport-radius}
\end{equation}
\end{theorem}
\begin{proof}
A set of fewer than $\tau_t(q)$ records is not a transversal, so it misses a
frontier and the transport theorem applies. A minimum transversal of size
$\tau_t(q)$ intersects every old frontier, ending that universal guarantee.
Every record set can be realized as a well-formed qualifier-only update, so
the argument does not require that those updates flip the target's bits.
\end{proof}
For positive record costs $w_\ell$, replace $|D|$ by
$\sum_{\ell\in D}w_\ell$ to obtain threshold $\tau_{t,w}(q)$. Every update
with cost strictly less than this threshold leaves an old cut, and a minimum
cost transversal attains the obstruction. Only for integer costs is the
corresponding integer budget radius $\tau_{t,w}-1$. An empty eligible-cut
family has no zero-update guarantee; a proof system admitting empty-frontier
constant certificates would instead need an infinite-threshold convention.
Neither case is silently assigned an ordinary finite radius.

\begin{proposition}[Guard-erasure inflation]
For each $n\geq1$, deleting one registered guard can inflate the certified
reuse threshold by a factor $n$, while all source hashes and cryptographic
primitives remain fixed.
\end{proposition}
\begin{proof}
Let $a_1,\ldots,a_n$ be supported-only leaves and $q=\bigvee_i a_i$. Its
inclusion-minimal local frontiers are $\{a_i\}$: one branch proves support and
zero opposition of the OR. Every valid frontier contains such a singleton, so
$\tau_t(q)=n$. Add an inactive catalogue defeater $g$ and define
$q'=G(q;g)$. Proving $p_{q'}=1$ requires $p_g=0$ and a positive branch;
proving $n_{q'}=0$ also requires $p_g=0$. The minimal frontiers are therefore
$\{a_i,g\}$ and $\{g\}$ hits all of them, giving $\tau_t(q')=1$.
The thresholds have ratio $n$; the radii are $n-1$ and zero, for which a ratio
would be undefined.
\end{proof}
The mechanism parallels obligation-erasure amplification in
Section~\ref{sec:erasure}: dropping a shared prerequisite before authentication
changes the resource being certified. In a guard star
$G(a;d_1,\ldots,d_k)$ with all defeaters inactive, every positive eligibility
cut must expose $a$ and every $d_j$. Omitting any $d_j$ leaves two completions
with identical disclosed records but different eligibility after activating
that defeater. Thus $k+1$ records are necessary, and the local cut attains this
bound. Authentication of $a$ alone cannot replace these negative obligations.

\subsection{From surviving cuts to witness capacity}
For a supported-only target $q$, let $\mathrm{src}_t(F)$ collect the source
principals of its opened leaves, deduplicating within that occurrence. Define
\begin{equation}
 \W_{t,\Delta}(q)=
 \{\mathrm{src}_t(F(c)):c\in\mathcal C_t(q),\
             \mathrm{status}(c)=(1,0),\ F(c)\cap\Delta=\varnothing\}.
 \label{eq:cut-generated-witnesses}
\end{equation}
The family is generated by positive, opposing and qualifier obligations in the
actual policy, rather than by an unauthenticated source union. It can be
exponentially large. An empty family for any required target means zero
complete-rendition reuse capacity; this case is handled before invoking the
nonempty-family planner of Section~\ref{sec:capacity}.

\begin{theorem}[Transport-certified source capacity]
Under the independent-compatibility assumption of Section~\ref{sec:model}, the
maximum number of complete renditions certifiable by unaffected old cuts with
source-obligation capacities $Q$ equals the configuration optimum
\eqref{eq:ilp} with families~\eqref{eq:cut-generated-witnesses} and the specified
target multiplicities.
\end{theorem}
\begin{proof}
Group each actual target occurrence by its transported cut's source footprint.
The resulting integer counts satisfy the configuration equalities and source
capacities. Conversely, each selected footprint in an integral feasible plan
has a corresponding old cut by definition. Transport supplies a current cut;
the compatibility assumption permits grouping occurrences into complete
renditions. Deduplication within an occurrence and separate charging of
repeated occurrences give exactly the coefficients in~\eqref{eq:ilp}.
\end{proof}
For fixed old state, $\Delta_1\subseteq\Delta_2$ only removes witnesses, hence
this certified-reuse capacity is nonincreasing. Recomputing a changed state's
inferences can instead increase or decrease its current witness family.
This theorem concerns a plan under available source capacities: it does not
transfer a previously issued permit. A transported interpretation or eligibility
cut is never issuer consent. The later protocol still requires every necessary
issuer's exact-scene permit and atomic budget admission. Checkpoint signatures
use a different message domain from grants and permits; promotion is rejected
even when a laboratory fixture uses the same key for two roles.

\paragraph{Relation to established methods.}
Neither negative provenance nor tracking updated data is new: $K$-relations
with difference, logical provenance and ProvSQL already cover substantial
parts of that problem~\cite{geerts2010,gradel2020,provsql2026}. Authenticated
answer completeness~\cite{pang2008}, persistent authenticated
dictionaries~\cite{anagnostopoulos2001}, dynamic hash accumulators and their
proof updates~\cite{dryja2019}, and incrementally verifiable computation such
as Nova~\cite{nova2022} establish that state authentication, incremental
verification and proof reuse have broad prior art. Revocable and use-once
proof-carrying authorization also predates this work~\cite{morgenstern2011}.
Our proposed contribution is the specific composition of typed polarity/guard
obligations, complete update cuts, repairable local-cut transport and
cut-generated source allocation, together with its certified-reuse functional
and guard-erasure analysis. It supplies inspectable conditional claims about
admitted scientific records. It cannot turn the authenticity of a record into
truth of an interpretation, completeness of all possible counterevidence, or
authorization by the record's source.

\section{SCENE: authenticating the counted object}
\label{sec:crypto}
The mathematical planner is not an authorization oracle. SCENE binds a chosen
realization to an exact structured manifest, then atomically consumes its
declared source obligations. It combines established components.
Macaroons already support contextual caveats~\cite{macaroons2014}; sanitizable
signatures~\cite{ateniese2005} and policy-based redaction~\cite{kissel2022} address
authorized editing and extraction. C2PA authenticates content-associated
assertions and ingredients and leaves trust interpretation to consumers~\cite{c2pa}.
SCENE neither implements those standards nor demonstrates a vulnerability in
them. It specializes provenance and consumable credentials to the capacity
semantics above.

\subsection{Authenticated manifest and disclosure closure}
A manifest contains a header (schema, renderer, cost rule and context), source
snapshot records, evidence records and explicitly identified object occurrences.
The reference scene supports square-pyramid geometry with integer millimetres,
positions, labels and colors. Source records hash packaged numeric source cards;
their URLs cite the publications, and separate metadata records the inspected
publication copies. These local hashes must not be mistaken for publisher signatures.

Canonical encoding accepts NFC strings, ASCII dictionary keys, safe integers,
booleans, null, lists and dictionaries. Floats, duplicate JSON keys, unpaired
surrogates and integers outside the exact JavaScript range are rejected; booleans
are rejected in integer fields. $\Enc$ emits compact UTF-8 JSON with ASCII keys
in lexicographic order, no insignificant whitespace, and integers in
$[-(2^{53}-1),2^{53}-1]$. Strings must already be NFC; the encoder rejects rather
than silently normalizes them. Object, source and evidence records are ordered
by stable IDs. The implementation's restricted encoding is documented as such,
not advertised as general JCS compliance~\cite{rfc8785}. Rendering resolves ties
by depth, object ID and face index, so record-list permutations cannot change
the committed geometry's deterministic rendering order.

Let $L(S)$ be that ordered record list. With RFC 9162 leaf/internal domains and
tree shape~\cite{rfc9162}, set
\begin{equation}
 \begin{aligned}
 t(S)&=\operatorname{hex}(\mathrm{MTH}(L(S))),\\
 B(S)&=\Enc\bigl(\{\texttt{tree\_root}:t(S),\
                  \texttt{size}:|L(S)|\}\bigr),\\
 R(S)&=\operatorname{hex}\!\left(H\!\left(
 \texttt{SCENE/scene-root/1}\Vert\texttt{0x00}\Vert B(S)\right)\right).
 \end{aligned}
 \label{eq:scene-root}
\end{equation}
Here the domain string is ASCII, $\texttt{0x00}$ denotes one NUL byte, and
$\operatorname{hex}$ is lowercase hexadecimal without a prefix. Thus both the
inner tree root and the returned scene root are serialized exactly as in the code.
Signing the count is necessary here: an RFC-shaped proof of leaf zero in a
three-leaf tree can verify under a claimed four-leaf size if an opaque sibling
digest is reinterpreted as a larger subtree. No collision is needed for that
proof-shape ambiguity. Our size-bound wrapper rejects it.

A disclosure for selected object IDs opens the header, each selected object,
its full transitive evidence closure, and every source record reached. The
verifier checks all inclusion proofs against the pinned wrapped root and count,
the exact external selection, record uniqueness and closure. Membership does
not imply that hidden real-world evidence is absent. Completeness is relative
to the authenticated adjacency lists and policy, which a trusted reviewer must
have admitted. A disclosure is not a new endorsed composition: producing a
strictly smaller projected manifest changes its root, except with negligible
collision probability, and requires its own authorization. Opening the entire
manifest or merely permuting its record lists does not create a distinct
committed scene.

\subsection{Grant, permit and atomic state transition}
The lab uses maintained Ed25519~\cite{rfc8032} and SHA-256. Pinned source keys live
outside untrusted packets. A source grant signs its ID, source ID, ledger ID,
evidence-subject digest, purpose, audience, renderer, cost rule, budget and
validity interval. The subject digest includes the complete source record and
that source's evidence snapshot, including support edges and qualifications.
For an instantiated scene the same source signs a permit containing grant and
ledger IDs, $R(S)$, context hash, recomputed cost and a request nonce.

The single authoritative SQLite ledger checks all required source permits in a
\texttt{BEGIN IMMEDIATE} transaction. It validates pinned keys, grant signatures,
scope, evidence snapshots, time, costs, nonces and balances. It debits all
sources and records all spent nonces, or rolls back every change. For each
grant $g$ the enforced invariant is
\begin{equation}
 Q_g^{(t+1)}=Q_g^{(t)}-c_g(S_t)\geq0,\qquad
 \sum_{t\in\mathrm{accepted}(g)}c_g(S_t)\leq Q_g^{(0)}.
 \label{eq:conservation}
\end{equation}
This is stipulated authorization accounting, with established consumable-credential
precedents~\cite{bowers2007}. Copying evidence or citing public heritage data is
not forbidden; the policy governs claims of the issuer's endorsement.

\begin{theorem}[Conditional scene and resource integrity]
Assume collision resistance of $H$, unforgeability of the signature scheme,
the specified injective encoding of record lists, honest admitted source/permit
issuance, externally pinned issuer keys, a trusted clock, correct
schema/closure/cost enforcement, and one intact serialized ledger. For selective
disclosure, also assume that the externally supplied expected root was produced
from a valid admitted manifest. In a committed transaction, the ledger recomputes
the supplied scene's root and context digest and compares them with each signed
permit; it does not trust a root claimed by the scene submitter. Then a permit
issued for one committed scene or context cannot authorize a distinct one except
by signature forgery or hash collision. Accepted transactions satisfy
\eqref{eq:conservation}, including object cloning costs.
\end{theorem}
\begin{proof}
If an accepting permit body was not signed under a required pinned key, it is a
forgery. Otherwise the issued permit fixes the root, context digest, grant and
cost. Acceptance of a different committed record list under that root yields a
collision in the wrapper or domain-separated tree; acceptance of a different
encoded context under its digest yields a collision. Different source evidence
snapshots are similarly excluded by their grant-subject digests. The resource
claim follows by induction over committed transactions: initial balances equal
signed budgets; each atomic step checks nonnegative residual capacity, records
previously unused nonces, and subtracts the recomputed vector. Object duplication
adds occurrence terms in~\eqref{eq:closure-cost}. A failing source causes no state
transition, so interleavings cannot create a partial multi-source debit.
\end{proof}
This is a standard conditional reduction and invariant proof; it does not
establish the safety of arbitrary renderer implementations, input policies or
maliciously issued endorsements.

\paragraph{State and display limits.}
Ledger events form a domain-separated hash chain. Receipts are \emph{not}
ledger-signed checkpoints: public verification of committed spending requires
an independently trusted pinned head or trusted ledger access. Auditing replays
signed grants and permits and checks balances, nonces and the chain; recomputing
scene costs also needs archived manifests. Database integrity and a single
authoritative state are explicit assumptions; hashes alone provide no consensus,
freshness or split-view prevention. Two disconnected offline verifiers starting
with the same unspent state can each accept a first spend, so universal offline
prevention is absent. The browser workbench marks edits as drafts and never
claims to verify endorsement. Camera controls are exploratory projections of
committed geometry; exact pixels, user attention and display frequency are not
authenticated.

\section{From symbolic envelopes to actual digital evidence}
\label{sec:realdata}
The empirical inputs include original museum photographs, an externally
acquired photogrammetric surface, earlier Anatolian source models and tables,
and archaeological-genetic supplementary records. Simple monument envelopes remain only a small,
fully transparent regression fixture. No image or surface in the corpus is
AI-generated, and no learned depth estimator is used.

\subsection{Corpus, archaeological reuse and source qualifications}
The image corpus contains 48 distinct museum objects: 33 from Cleveland and 15
from the Metropolitan Museum, with 48 different original-image hashes. Museum
catalogue attribution groups 32 as Egyptian and 16 as Sudan/Nubia. The selection
covers 11 investigator-defined material/form domains, including inscribed stelae,
reliefs, stone sculpture and tools, wooden figures and coffins, pottery, glass,
faience, textile, beadwork and metalwork. These are heterogeneous source objects,
not 48 independent excavations or a representative cultural sample. Every
bundled web JPEG has explicit CC0/public-domain support in its museum record
under the museums' open-access programs~\cite{cmaopen,metopen}; the 48 original
files total 11,151,616 bytes. Separate metadata snapshots retain accession,
medium, date, provenance, dimensions, bibliography, rights and exact file hashes.

Twelve Sudan/Nubia objects have Faras as their catalogued site. Four others
retain qualified or broad attribution: the Napatan pendant's findspot is
unspecified, the Taharqo shabti is probably from Nuri, a ram's-head amulet is
possibly Nubian, and the kneeling king has a regional Egypt/Sudan attribution.
No cultural label is promoted into a certain excavation location.

Two actual objects provide the relevant historical depth. Cleveland's pendant
1987.1 is described as a lion-head component, probably Egyptian and most likely
an earlier gaming piece, adapted as a Napatan amulet~\cite{cmapendant}. Its page
and metadata retain the museum's differing chronological descriptions rather
than resolve them by guesswork. The Met's kneeling Kushite king, 2002.8, is
catalogued with later reduction of the double uraei and erasure of neck amulets
and the royal name~\cite{metkushite}. These attributed interpretations document
material reuse and selective alteration. They motivate source, transformation
and qualification records; they do not establish the motives of every ancient
viewer or a universal causal law connecting ancient and modern media.

\subsection{Actual photogrammetry and computed surface depth}
The SEAMS dataset for World Museum Liverpool 56.22.225, a wrapped votive cat mummy
from Thebes, provides the original OBJ, material, texture and README under
CC BY 4.0~\cite{pubblico2026}. Our parser finds 133,414 vertices and 266,844
triangles. The source reports 181 aligned cameras. The actual four-file deposit
does not include the raw camera images or project despite its generic README
listing them, so we reproduce \emph{downstream analysis}, not original SfM.
The source units are reported as metres, but focal-length descriptions disagree
(50 versus 55 mm), the README reports GCP RMSE 0.185771 m with no check points,
and we do not independently validate physical scale.

The original coordinates and connectivity are retained. An edge-incidence
analysis finds no boundary or nonmanifold edges and five connected components;
these are computational properties of this mesh, not evidence that the physical
object has five independent parts. PCA supplies a recorded exploratory frame.
For each of its three axes and both viewing directions we project actual
vertices into a $512\times512$ orthographic grid and retain the nearest sampled
depth, with a $3\times3$ square splat and an explicit occupancy mask. Let $q_i$
be the rounded projected pixel coordinate of vertex $v_i$, and let $u$ be the
signed unit viewing axis. In symbols,
\begin{equation}
 \begin{aligned}
 \mathcal I_{ab}&=\{i:\|q_i-(a,b)\|_\infty\leq1\},\\
 d_{ab}&=\max_j u^T(v_j-\mu)
         -\max_{i\in\mathcal I_{ab}}u^T(v_i-\mu),\\
 M_{ab}&=\mathbf{1}[\mathcal I_{ab}\ne\varnothing].
 \end{aligned}
 \label{eq:pointdepth}
\end{equation}
The depth expression applies only when $M_{ab}=1$; otherwise the stored depth
is NaN. The global-minus-local maximum gives the nearest sampled distance from
the view-relative origin in either viewing direction. Empty pixels
include both background and sampling gaps; they are not automatically scan holes.
These are computed views of a reconstructed surface, not sensor depth, dense
triangle rasterization, CT, internal anatomy or calibrated metric measurements.

A deterministic 2,000-vertex sample supports covariance analysis over 32-point
neighborhoods, including the query vertex.
For covariance eigenvalues $\lambda_{i,1}\leq\lambda_{i,2}\leq\lambda_{i,3}$,
we report the dimensionless local roughness proxy
\begin{equation}
 r_i=\frac{\lambda_{i,1}}{\lambda_{i,1}+\lambda_{i,2}+\lambda_{i,3}},
 \qquad L=I-D^{-1/2}WD^{-1/2}.
 \label{eq:surface-descriptors}
\end{equation}
Here $L$ is the normalized Laplacian of a symmetrized eight-nearest-neighbor sample graph
with Gaussian weights; its first 32 computed eigenvalues are archived. The
roughness, unsigned neighboring-normal variation and spectrum depend on vertex
density, sampling, bandwidth and reconstruction. They are not calibrated
curvature, dating evidence or cryptographically unique shape identifiers.

\subsection{Image analysis without invented depth}
Each photograph is analyzed at long-side limits 128, 256 and 512 pixels using
grayscale entropy, Sobel-gradient summaries, smoothed structure-tensor coherence
and high-frequency Fourier energy. All 144 rows retain parent object IDs and
analysis scales. They describe image texture and edge structure; an RGB image
alone supplies no metric depth in this method. As a concrete information-loss
control, reversing the decoded pixel order preserves every channel histogram
for all 48 objects while changing all 48 exact RGB hashes. This demonstrates
why visual descriptors cannot replace the authenticated spatial representation.
The observation is algebraic and corpus-checked, not a trained authenticity classifier.

\begin{figure}[t]
\centering\includegraphics[width=\linewidth]{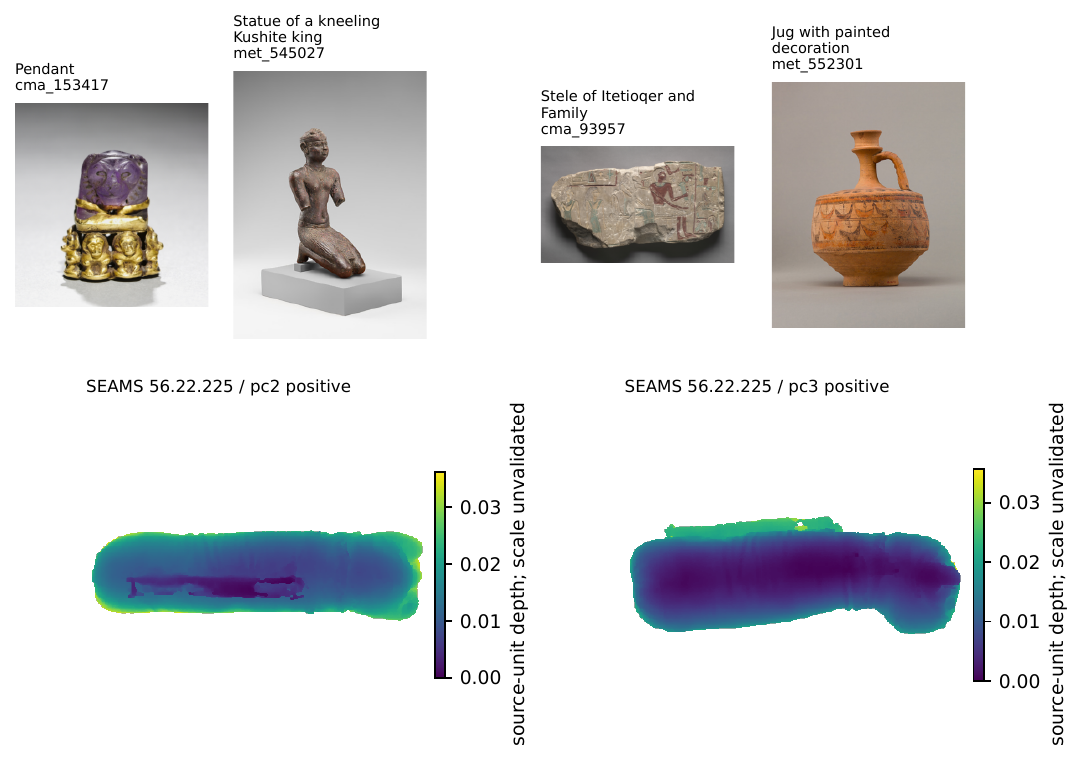}
\caption{Actual source objects and computed surface views. Upper row: Napatan
pendant, kneeling Kushite king, Egyptian stela and Faras jug (Cleveland and Met,
CC0; catalogue IDs identify full source records). Lower row: two orthographic
point-depth projections of SEAMS 56.22.225 (Pubblico / National Museums Liverpool,
CC BY 4.0), derived with new cameras and masked gaps. These are surface-model
projections, not measured depth photographs or validated physical scale.}
\label{fig:material-evidence}
\end{figure}

\section{Cryptographic region extraction over real pixels and geometry}
\label{sec:asset-security}
The fixed-scene protocol is extended by two typed evidence layers. Whole-file
hashes identify original inputs; structured commitments verify the \emph{actual
disclosed content}, its selection coordinates and mandatory qualifications.
No institutional signature is fabricated: every experiment uses an explicitly
public local test key. Its private seed is also published for reproducibility:
these fixtures test implementation, not unforgeability against a reader who
knows that seed. Real issuance requires fresh protected private keys.

\subsection{Image tiles, scale and off-crop context}
The image manifest signs original-file and decoded-RGB hashes, dimensions,
decoder and resampling versions, exact rational scale maps, source identity,
scene/context bindings, calibration status, evidence graph and a count-bound
tile commitment. The reference pipeline uses EXIF orientation, RGB conversion
without an ICC color transform, and sequential BOX half-sampling with signed
level dimensions. Tiles are $64\times64$ except at boundaries. A leaf binds
its level, row, column, bounding box, source and RGB-buffer digest.

An independently fixed policy specifies a half-open crop, level, admitted
manifest, required evidence and mandatory tile IDs. Supplemental tiles may lie
outside the visible crop: for example, context needed to interpret a local
detail can remain mandatory. A proof supplies actual RGB tile buffers and paths.
The verifier checks the signature, expected policy, count-bound tree, coordinates,
pixel hashes and complete evidence closure, reconstructs the crop from verified
tiles and requires equality to the supplied crop pixels. An optional original-file
check also verifies the encoded file, decoded base image and actual derived
levels. Rebuilding expected policy from the attacker's packet is forbidden.

The scene first commits original asset hashes, then its final root is bound by
the signed asset manifest. Committing that whole asset manifest back into the
same scene would create a dependency cycle and is not the construction. Byte
integrity, decoded-pixel equivalence and perceptual similarity are separate
policies; this prototype enforces exact equality within the declared pipeline.

\subsection{Triangle blocks and vertex dependency closure}
The mesh layer commits little-endian float64 vertex buffers and uint32 triangle
buffers in blocks of 1,024 vertices and 2,048 triangles. A signed, count-bound
tree authenticates 131 vertex blocks, 131 triangle blocks and one metadata header.
That header binds the original OBJ hash, units, calibration caveats, analysis
frame and context. A disclosure opens exactly the independently selected triangle
blocks plus \emph{all vertex blocks referenced by those triangles}. Missing or
extraneous blocks, changed coordinates, rewired indices and altered metadata fail.
This is reference closure, not a claim that all geometrically adjacent faces
or all historical evidence have been disclosed.

\begin{proposition}[Conditional regional extraction integrity]
Fix the trusted issuer key, expected admitted root/manifest, selection and
qualification policy before receiving a packet. Under signature unforgeability,
hash collision resistance, injective descriptors and correct deterministic
decoding, a verifier cannot accept changed disclosed pixels or mesh buffers,
different selection coordinates, or omitted mandatory dependencies as that
authorized extraction except by forgery or hash collision.
\end{proposition}
\begin{proof}
An accepted signed body outside the honest issuer's issuance record gives a
forgery. An honestly issued alternative with a different expected manifest
digest or root fails the fixed policy unless a hash collision occurs; an
identical canonical binding is the same admitted commitment. Under that
commitment, changing an opened
descriptor without changing its root yields a tree/wrapper collision. Changing
buffer bytes while retaining the authenticated digest yields a buffer-hash
collision. Exact selection comparison and deterministic grid/block indexing
exclude positional or scale relabeling. Finally, required tile IDs or vertex-block
IDs are recomputed from the fixed selection and authenticated dependency data;
the verifier demands set equality, so omission is rejected rather than inferred
from absent unauthenticated records. The image verifier additionally reconstructs
the crop from checked buffers and compares the supplied pixel output. The mesh
checker authenticates disclosed buffers; extending that guarantee to a visible
surface requires a renderer to consume those checked buffers.
\end{proof}
These proofs concern exact evidence extraction. They provide no confidentiality,
capture authenticity, robustness to benign edits, historical truth, institutional
consent or stateful spending guarantee on their own. Spending remains the
separate serialized authorization layer; an archived membership proof is reusable.

\subsection{Real-corpus authorization and attack workload}
Across the 48 image objects, three resolution levels produce 9,103 authenticated
tiles. Four region selections per level give 576 actual-image controls, all
accepted with original-file checking. Nine mutation families generate 5,184
unauthorized packets, all rejected: crop and tile pixel edits, coordinate and
source relabeling, crop-policy transplantation, evidence/qualifier removal,
mandatory tile omission, context editing and wrong-key checks. Their case-level
records identify the parent object, crop, level, proof size and outcome; they are
constructed tests of exact rules, not independent samples from real attackers.
Additional valid-proof transplantation experiments isolate selection binding
from mere rejection of malformed signatures. Across three real region selections
per image, 144 original controls and 144 honestly signed alternate-context controls
are accepted under their own policies. Under the original fixed policy and key,
864 substitutions are rejected: valid other-asset or other-level proofs, the
wrong actual source file, foreign-key signatures, honestly signed different-scene
claims, and histogram-preserving swaps of distinct $8\times8$ image patches.
All 144 patch pairs differ spatially while preserving their joint RGB histogram.
These cases specifically show that a valid signature or preserved photometric
summary does not authorize a different extraction.

On the real mesh, 16 actual triangle-block selections are accepted; 160 buffer,
topology, dependency, context and calibration mutations plus three external
selection/metadata transplants are rejected. Block closure reveals a practical
cost: original indexing scatters selected triangles' vertices across many blocks.
This motivates the explicit layout experiment below.

The capacity equations also apply to the 49-object image/SEAMS corpus without
pretending that distinct artifacts are historical substitutes. Each occurrence
requires both a snapshot-admission role and an interpretation-review role.
One declared review round uses counts $(33,15,1,49)$ for Cleveland-snapshot,
Met-snapshot, SEAMS-snapshot and interpretation review. All are \emph{local study
resource counters}, not museum-issued endorsement rights. Twelve budget cases
agree with the fixed-witness capacity equation and integer optimization. At
one round per role, keeping the conjunction gives capacity one; erasing it
gives two. The corpus demonstrates the implemented accounting semantics; the
unbounded-family theorem remains a mathematical construction.

\subsection{Engineering the closure cost without changing geometry}
For an original spatial target face set $T$, a layout induces its covering
triangle blocks $\mathcal B_F(T)$ and the vertex-block dependency closure
$\mathcal B_V(T)$. The base64 buffer contribution to proof size, before JSON
framing and other fields, is exactly
\begin{equation}
 B_{\rm buffers}(T)=\frac43\left(
 12\sum_{F\in\mathcal B_F(T)}|F|+
 24\sum_{V\in\mathcal B_V(T)}|V|\right),
 \label{eq:proof-cost}
\end{equation}
with descriptors, metadata, authentication paths and JSON framing counted
separately. Twelve bytes encode one uint32-indexed triangle and 24 encode one
float64 vertex; both raw lengths are multiples of three, so these base64 buffers
need no padding. The security rule
fixes complete reference closure; ordering controls how much unrelated material
the blocks additionally reveal. Reducing that cost does not justify deleting
required dependencies.

We compare original indexing with deterministic Morton ordering~\cite{morton1966} of vertices and
triangle centroids, using 21 interleaved bits per coordinate for ordering only.
Original coordinate bits, connectivity and winding are unchanged. Signed layout
metadata commits archived bijective original-ID maps. Sixteen nonempty cells of
a $4\times4\times4$ original-coordinate centroid grid supply identical target
face-ID sets for both layouts. The complete-map experiment verifies the target
predicate and exact geometry equivalence; the lightweight disclosure verifier
checks pinned block covers and vertex closure, not a succinct proof of the
geometric centroid predicate.

Median packet size falls from 7,417,431 to 842,247 bytes; the median of paired
reduction factors is $9.13$. Median referenced vertex-block count falls from
131 to 20.5, and median triangle overfetch from 36.37 to 1.55 times the requested
face count. Original-ID maps and ROI target IDs are shared one-time data excluded
from these per-packet sizes and retained in the archive. This is a measured
layout improvement on one real mesh, not a universal compression bound, runtime
speedup or new cryptographic primitive.

\begin{figure}[t]
\centering\includegraphics[width=\linewidth]{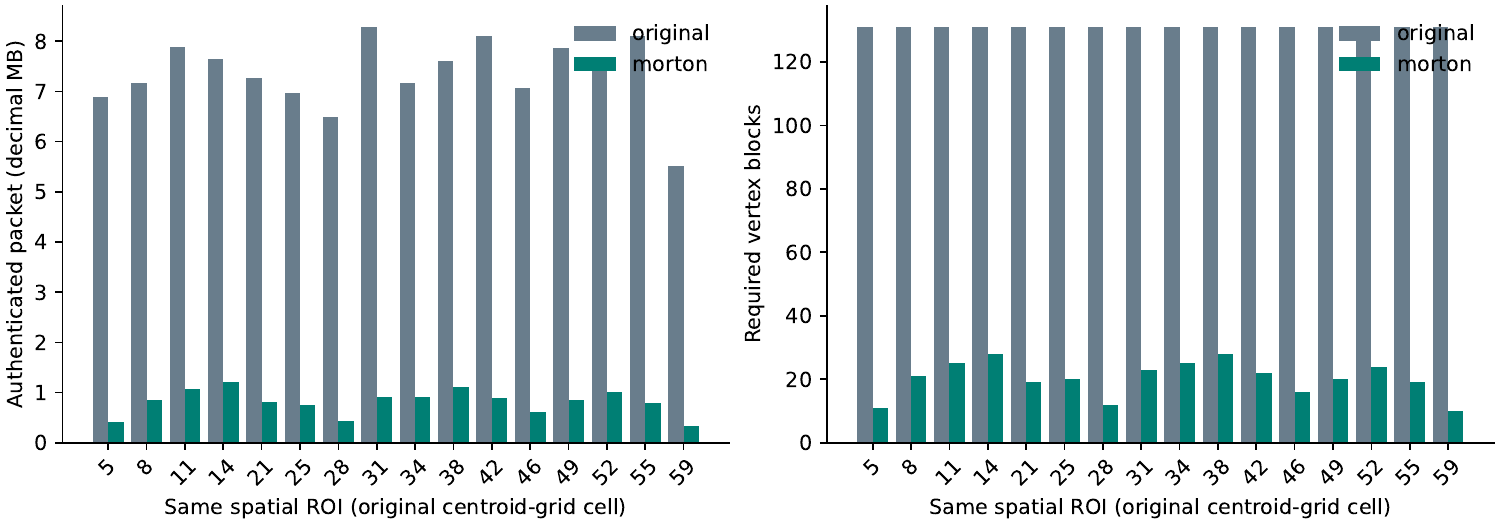}
\caption{Same-region extraction on the actual photogrammetric mesh before and
after spatial reindexing. All coordinate bits and triangle incidence remain
identical under archived ID maps. Reduction concerns complete serialized proof
packets; one-time maps are an additional shared cost.}
\label{fig:mesh-locality}
\end{figure}

\section{Earlier archaeology, structured models and qualified social claims}
\label{sec:early}
The expanded inputs distinguish geometry, excavation-related tables, a formal
ontology and genomic/osteological classifications. Each supports a different
kind of reconstruction. Their record counts cannot be summed into a population
sample, and physical/model resolution does not establish interpretive certainty.

\subsection{G\"obekli Tepe beyond monument outlines}
Dietrich et al.'s early-Neolithic study concerns cereal processing at G\"obekli
Tepe in the tenth/ninth millennia BCE~\cite{dietrich2019}. We retain all twelve
original supplementary workbooks unchanged, with source URLs, file hashes and
CC BY 4.0 attribution. Read-only XML extraction preserves worksheet/cell
coordinates, original numeric lexemes, formula caches and explicit missing
categories. The workbooks contain 3,374 nonempty source cells. The S1 aggregate
is 7,268 tools: 3,351 handstones, 1,945 grinding bowls/slabs and 1,972 pestles.
This aggregate differs from the 821 distinct find numbers in detailed inventory
subsets; neither is a count of raw tool scans. Overlapping subsets are not added.

The phytolith table supplies 48 morphotypes across nine samples, yielding 432
actual count cells. Independent column sums reproduce all nine published totals,
from 290 to 1,210. Figure~\ref{fig:source-conditioned} exposes the full matrix,
including recorded zeroes, instead of substituting a schematic monument.
These observations support a source-specific analysis of plant processing;
the tables do not record the operators' sex, gender or authority. Any gendered
labor claim would need an additional, admitted inference and evidence.

The sources also expose an unresolved join: S10 worksheet 1, K3 records
\texttt{M12-133}, whereas S12 A1212 records \texttt{M13-133}. A second comparison
uses \texttt{M11-133} in S8 A8 and \texttt{M13-133} in S6 A7. The strings remain
unchanged. A valid digest on both records cannot establish that either pair
refers to the same specimen. An alias relation must be a separate qualified
interpretation record; strict equality cannot be replaced by an unaudited
normalization rule. This is a real example of the boundary between source
integrity and admissible inference.

Lourentzaki, Papadopoulou and Roche's G\"obekli ontoterminology supplies a
machine-readable model of pillar finds and their named relations~\cite{lourentzaki2024}.
The original 274,754-byte RDF is retained without alteration. Parsing its RDF/XML
syntax yields 3,732 unique triples, 575 subjects and 162 explicitly declared
named individuals. These syntax counts are different from the repository's
Prot\'eg\'e axiom/individual metrics; anonymous nodes and reasoning conventions
must not be conflated. The ontology is an authored conceptual model, not raw
excavation geometry or complete world knowledge. Its DOI registry specifies
CC BY-NC-ND 4.0; the ancillary license explicitly isolates this unchanged source
from the other data licenses. We distribute syntax statistics and inspection
code rather than a modified complete ontology.

\subsection{Original \c{C}atalh\"oy\"uk artifact models}
The Perry--Lara--Palma deposit provides thirteen artifact groups and fourteen
original mesh files, including two variants of one skull figurine~\cite{perry2020}.
Its twelve OBJ members and two binary FBX geometries retain exact publisher
bytes, with reversible storage recipes for OBJ and gzip wrappers for FBX:
decompression recovers exact original archive members. The complete collection
has 495,152 stored vertices and 628,600 original polygons, which yield 990,248
explicit fan-triangles. Excluding the additional skull variant gives thirteen
groups, 490,985 vertices and 981,918 triangles. These are source-published
\emph{printing-ready processed models}; acquisition uncertainty, physical units
and site georeferencing are unreported. Their dense meshes are not evidence of
raw-scan fidelity or metric accuracy.

The analysis preserves the original polygons, coordinate bytes and source
hashes, records the fan triangulation, and checks boundary/nonmanifold edges,
components and zero-area faces. Flint and horse each contain one exactly
degenerate fan-triangle and respectively one and two nonmanifold edges; the
stamp has three boundary edges and two nonmanifold edges. These defects are
reported without repair. An analytical triangulation does not certify convexity,
planarity or absence of self-intersections in every source polygon. FBX local
geometry is kept distinct from scene-node transforms.

Vertex-weighted PCA supplies documented exploratory cameras and independent
per-object display scales. Three 256-pixel point-splat projections per model
produce 42 computed views with explicit occupancy masks. We also retain
scale-free edge-length, area and radial quantiles, PCA shape fractions and
component-level summaries. Model-level aggregation gives each of thirteen
artifact groups equal weight; it does not treat vertices as independent samples.
The median first-PC variance fraction is 0.6791 (range 0.3486--0.9936).
These describe the supplied representations rather than historical function.
Figure~\ref{fig:early-cohort} shows the complete collection. The source filename
\texttt{mother-goddess} remains a quoted cataloguing interpretation, not a finding
about the depicted person's identity or social system.

\begin{figure}[p]
\centering\includegraphics[width=\linewidth]{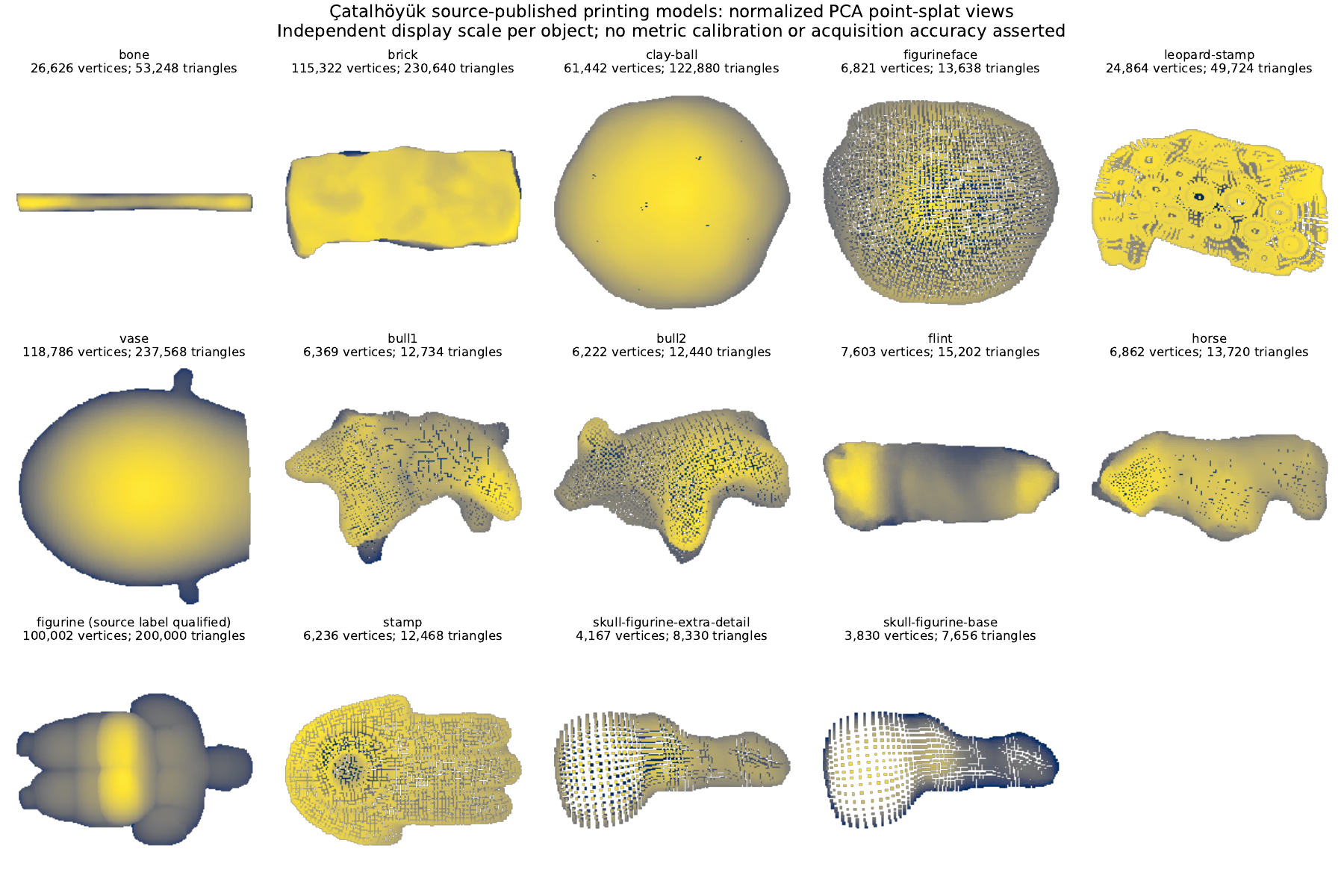}
\caption{All fourteen original \c{C}atalh\"oy\"uk model files, including two skull
variants (Perry, Lara and Palma, CC BY 4.0). Computed PCA point-splat views use
independent display scales and explicit sampling gaps. Labels report source
model complexity; neither color nor mesh density encodes acquisition confidence.
No physical calibration or gender-role inference follows from this display.}
\label{fig:early-cohort}
\end{figure}

\subsection{Social interpretation as an evidence obligation}
Archaeology can test simple accounts of past social organization, but the tested
variables matter. Y\"unc\"u et al. analyze 131 East-Mound \c{C}atalh\"oy\"uk
paleogenomes and report changing kinship patterns, maternal connections and
sex-associated differences in subadult gifts~\cite{yuncu2025}. Their analyses
provide no indication of sex-biased mobility. Co-burial, maternal relatedness,
mobility, gender and political authority remain separate propositions. The
2026 v4 data record warns that earlier relatedness-simulation files were
incorrect~\cite{cataldata2026}, making version-pinned interpretation especially
relevant. The large chromosome-list archives are not burial/gender tables.
We did not acquire the paper's original S1--S28 individual tables and therefore
report this as literature evidence, without claiming a reproduced 131-person
analysis or bundling unlicensed sequence data.

For an actual larger social-data fixture, the artifact retains Cassidy et al.'s
original CC BY 4.0 supplementary workbook~\cite{cassidy2025}. It covers Iron Age
Britain and European comparison samples; it is not an Anatolian substitute.
Six source-coordinate-preserving extracts contain 2,405 comparative sample rows
with 2,404 distinct source IDs, 1,485 complete WBK pair classifications
($\binom{55}{2}$), 386 selected within/between-site relationships, 156 site/group
aggregates and two 57-individual tables. The duplicated source ID
\texttt{StPet2\_MIX} is retained at both original rows. The selected 386
relationships are not a complete population-wide matrix; an unreported pair is
not a measured unrelated pair.

The 57-individual summary records 24 \texttt{XX} and 33 \texttt{XY} molecular
classifications. Osteological source categories are six \texttt{Female}, five
\texttt{Male}, one \texttt{Male?} and 45 \texttt{-}. All 45 unreported labels
remain visible. One \texttt{XX} record carries an osteological \texttt{Male}
label and another \texttt{Male?}; these are source classification differences
requiring contextual review, not discoveries of gender identity. Neither
matrilocality nor a genetic classification directly establishes legal office,
political control or universal matriarchy. The method allows such an interpretation
to be represented and challenged by explicit premises, while preventing its
presentation as an already measured source field.

\begin{figure}[p]
\centering\includegraphics[width=\linewidth]{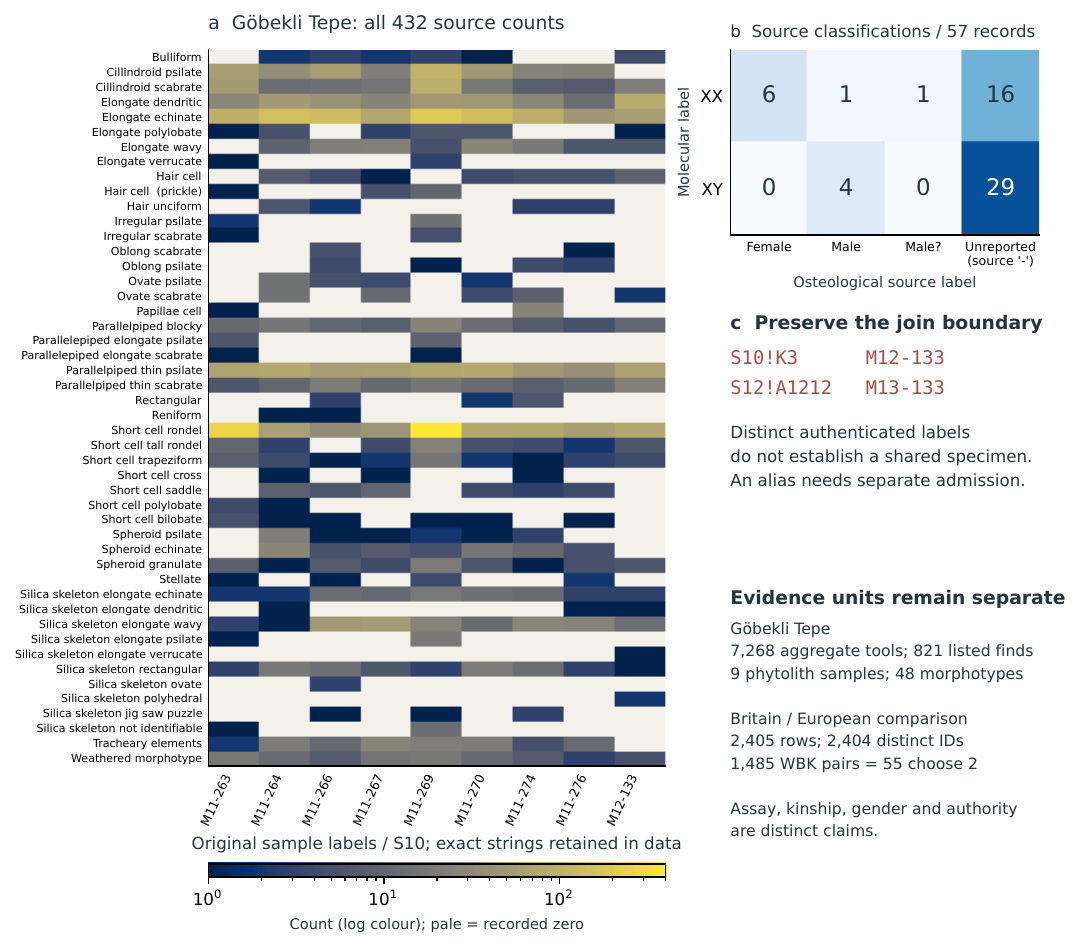}
\caption{Source-conditioned quantitative evidence. (a) All 432 G\"obekli phytolith
counts; logarithmic color separates positive counts from recorded zero. All nine
column totals match the source. (b) Original molecular/osteological categories
for 57 records in the separate British dataset, with missingness retained.
(c) An unresolved cross-table identifier difference. Authentication establishes
the recorded values, while a join or social interpretation needs its own rule.}
\label{fig:source-conditioned}
\end{figure}

\subsection{A 2026-informed, offline research interface}
The accompanying atlas links exact source cells, the phytolith matrix,
model projections, relationship classifications and a provenance/detail pane.
Selecting an item reveals its original coordinate or file hash and the applicable
qualifications. Separate acquisition, transformation and interpretation records
keep a computed projection distinct from a measurement. Detail is progressively
revealed without hiding the immediate numeric view or requiring motion.

This design draws on actual 2026 work: Wiebel, Brender and Eckardt's linked
annotation/graph/attribute views for Roman millefiori artifacts~\cite{wiebel2026},
and LLUVIA's staged uncertainty and level-of-detail workflow~\cite{luque2026}.
Linked views and uncertainty layers are established visual-analysis techniques;
the atlas adds source-coordinate bindings to this paper's verification workflow.
It uses local HTML/DOM/SVG and original data, without a learned embedding or
external runtime. These references guide interface construction, not a claim
that their user-study effects transfer to archaeology or that this atlas has
been validated with participants. All published figures are reproducible static
scientific plots; interactive views are separate ancillary files.

\section{Regional archaeology as a source of comparison constraints}
\label{sec:regional}
The next research object is a qualified archaeological contrast. A comparison
can change even when every disclosed number is genuine: the selected deposits,
counting unit, taxon mapping, denominator or weighting may have changed. We
therefore extend the source investigation to Iran, Transjordan, Oman/UAE,
Lebanon and China, using original specimen exports and supplementary files.
Persian, Arabic and Chinese queries and their actual outcomes are recorded in
the artifact. Local-language institutional records supply useful excavation and
terminology checks; a search result is not represented as an acquired database.

The historical question is deliberately narrower than a universal explanation
of pig avoidance. Absence from routine human dairying does not imply absence
of economic value or identify a prohibition's cause. Price and Corcoran-Tadd
contrast the reproductive, feeding and mobility properties that make livestock
useful to households and institutions~\cite{pricecorcoran2026}; Price's
Mesopotamian synthesis also distinguishes dietary importance from administrative
visibility~\cite{priceswine2026}. At Xipo, isotope evidence supports intensive
pig husbandry within a cereal economy~\cite{zhangxipo2021}. These are substantial
antecedents, not new discoveries of this paper. Ecological costs, household
production, institutional procurement and culturally specified avoidance can
coexist. Their relative causal importance requires dated, context-compatible
evidence; a modern country's predominant religion is not a treatment variable
for a prehistoric assemblage.

\subsection{A source-specific archaeological hypothesis at Chogha Mish}
\label{sec:chogha}
The original Open Context export for Chogha Mish, southwestern Iran, contains
16,824 distinct specimen URIs and 59 flat fields~\cite{choghamish2010}.
The project's description of more than 30,000 historical specimens does not
become the size of this acquired table. Three source predicates share the label
\texttt{Number of fragments}; the flat columns have 16,681, 114 and 29 populated
rows respectively. Each row has exactly one such value. Source-specific
coalescing gives a recorded-field sum of 16,825; using only the first column omits 143 records.
The separately published 4,705-row Protoliterate table is wholly included in
the full export, which has 4,819 Protoliterate rows. They are overlapping releases,
not two independent samples.

For an exploratory comparison, fix the source phase \texttt{Protoliterate},
source contexts \texttt{Room floor} and \texttt{Occupational Debris}, and the
five exact taxon labels \texttt{Sus sp.}, \texttt{Bos taurus}, \texttt{Ovis/Capra},
\texttt{Ovis sp.} and \texttt{Capra sp.}. The numerator is the recorded-field total of
\texttt{Sus sp.}; the denominator counts these five selected labels. It does not count all excavated material, unique physical fragments, independent
animals or meat weight. The primary chapter states that jaws were recorded
repeatedly, one line per tooth, before merging in its published analysis
(p.451, footnote4); no validated merge/refit crosswalk is present here. The
12 floor and three debris Sus URIs are fifteen source records, not fifteen
independent bones or animals. For example, the maxilla record CGM12288 and
upper-right tooth records CGM12289--12293 cannot establish five additional jaws.
The fragment-field sum above must therefore not be promoted to a unique-piece
count~\cite{levtov2017}.
\texttt{Sus sp.} does not distinguish wild boar from domestic pig. The resulting
source-coded contrast is
\begin{equation}
 \widehat\Delta_{\rm recorded}
 =\frac{12}{95}-\frac{3}{271}=\frac{2967}{25745}\simeq0.11525.
 \label{eq:chogha-recorded}
\end{equation}
The preferred analysis excludes the source body-part group \texttt{Head}, leaving $5/61$ versus $2/164$,
with difference $349/5002\simeq0.06977$. Thus the direction is not confined to
repeated jaw/tooth records. The seven postcranial Sus URIs have no validated
individual/refit links and do not establish seven animals. Postcranial restriction
alone does not remove recovery,
transport or depositional effects. Figure~\ref{fig:regional-discovery} displays
both the source contrast and the exact anatomical strata.

\begin{figure}[t]
\centering\includegraphics[width=\linewidth]{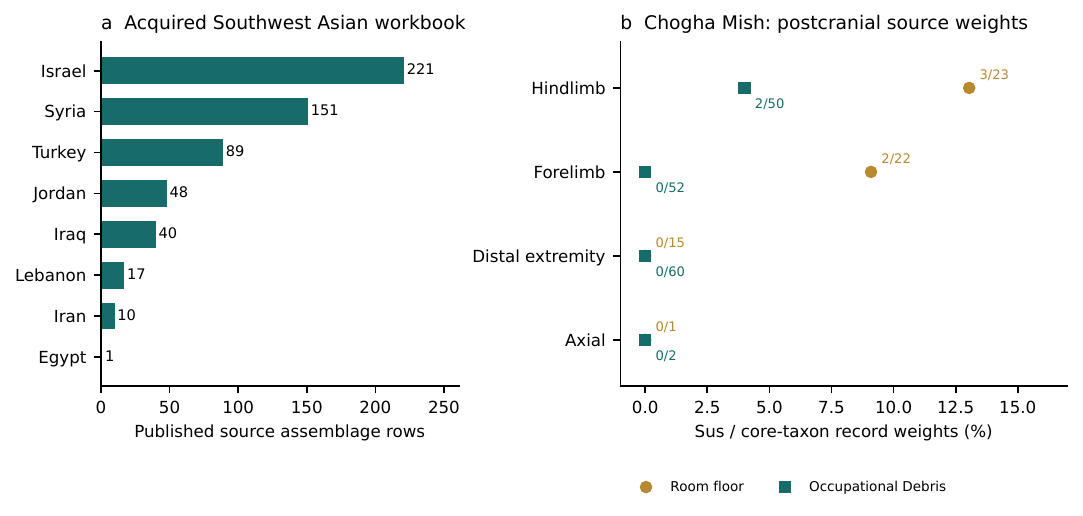}
\caption{The left panel inventories the original regional workbook's source rows
by country label. The right panel shows the exploratory source-coded Chogha Mish
comparison. Exact taxon, phase, context and fragment-field rules define the
denominator. Anatomical strata
expose the support of the floor/debris contrast; the values count recorded-field weights; head records can repeat a jaw once per tooth. The
postcranial comparison is primary, and neither comparison counts independent animals. The proposed archaeological test requires
matched dating, depositional review and controlled recovery.}
\label{fig:regional-discovery}
\end{figure}

The prospective hypothesis is that \emph{Sus enrichment of floor-context
assemblages persists within independently dated coeval contexts after matched
recovery/sieving and restriction to comparable postcranial material}. It would
be weakened or refuted if matched recovery or revised context assignment removes
or reverses the contrast. Persistence would challenge a purely context-mixture
explanation but would still leave animal transport, disposal and household
management as alternatives. Useful new observations include screened sediment
volumes, recovery mesh sizes, refitting fragments, element density, cut marks,
direct phase associations and osteometric wild/domestic evidence.

This is hypothesis generation after inspection of the existing source. The independent source-transfer audit below tests its scope of transport;
no independent confirmation of the exact floor-use hypothesis is established. The
source phase has an editor-assigned range of 3600--3200 BCE; pooled legacy
context types are not demonstrated individual households or exactly simultaneous
deposits. All fifteen selected Sus records share a phase-level context URI;
it is not a unique room/deposit identifier. Their editor-assigned dates establish
phase compatibility, not coeval deposition. The primary chapter reports no
sieving on p.446; mesh size, screened sediment volume and individual deposit
associations cannot be reconstructed from this export. Fragments can derive from the same bone or animal, so an independence-based
significance test would overstate the evidence. The original 2017 study already
reports broad changes in pig representation and explicitly sets aside more
specific spatial analysis~\cite{levtov2017}. The focused context/body-part contrast
was not located in the checked chapter, but that check is not an exhaustive
priority proof. Its warning that Sus may include wild boar remains applicable.

\subsection{The shared weighting constraint exposed by anatomy}
\label{sec:shared-weights}
Let $a_h$ and $b_h$ be the two recorded Sus shares in a common anatomical stratum
$h$, and write $\delta_h=a_h-b_h$. A comparison standardized to one declared
anatomical distribution uses \emph{the same} weights on both contexts:
\begin{equation}
 \Delta(w)=\sum_{h=1}^{H}w_h(a_h-b_h),\qquad
 w\in\mathcal S_H=\{w\geq0:\textstyle\sum_h w_h=1\}.
 \label{eq:shared-comparison}
\end{equation}
These weights define an estimand. Equal stratum weights do not estimate an
unknown ancient carcass distribution: the preferred four-stratum analysis
assigns 25\% to an axial stratum supported by only one versus two recorded fields.

\begin{proposition}[Common-weight comparison bounds]
For fixed nonempty matched strata,
\begin{equation}
 \inf_{w\in\mathcal S_H}\Delta(w)=\min_h\delta_h,
 \qquad \sup_{w\in\mathcal S_H}\Delta(w)=\max_h\delta_h.
 \label{eq:common-bounds}
\end{equation}
If the shared-weight restriction is erased and independent simplex weights
$u,v$ are permitted, the lower bound becomes
\begin{equation}
 \inf_{u,v\in\mathcal S_H}
 \left(\sum_hu_ha_h-\sum_hv_hb_h\right)
 =\min_h a_h-\max_h b_h.
 \label{eq:uncoupled-bounds}
\end{equation}
\end{proposition}
\begin{proof}
A convex combination of the $\delta_h$ lies between their minimum and maximum;
placing all mass on an extremizing stratum attains either endpoint. For independent
weights, minimize the first combination and maximize the second separately.
Both choices are simultaneously feasible only after their coupling is removed.
\end{proof}

For the preferred postcranial analysis, in source order Axial, Distal extremity,
Forelimb and Hindlimb, the observed recorded-field shares are
\begin{equation}
 a=\left(0,0,\frac{2}{22},\frac{3}{23}\right),\qquad
 b=\left(0,0,0,\frac{2}{50}\right).
 \label{eq:chogha-strata}
\end{equation}
Every shared nonnegative weighting preserves a weakly nonnegative contrast;
strict positivity requires positive weight on Forelimb or Hindlimb. The exact
interval is $[0,1/11]$, and common quarters give $1147/25300\simeq0.045336$.
The raw head-including audit is separately available, with interval
$[0,715/3638]$; repeated tooth records prevent interpreting it as a unique-piece
comparison. In the preferred analysis, coupling erasure lowers the endpoint
to $-1/25$, obtained by contrasting a zero-Sus floor stratum with debris
Hindlimb. The lower-bound loss remains exactly four percentage points.
That endpoint answers a different anatomical question although no source value
has changed.

More generally, let $\mathcal F$ be a nonempty set of admitted joint assignments
of mappings, weights or shared measurements, let $\mathcal E(\mathcal F)$ be a
relaxation formed by deleting specified consistency constraints, and suppose
both infima below are finite. Define the \emph{constraint-erasure lower-bound gap}
\begin{equation}
 G_\theta(\mathcal F,\mathcal E)
 =\inf_{z\in\mathcal F}\theta(z)
  -\inf_{z\in\mathcal E(\mathcal F)}\theta(z)\ \geq\ 0,
 \qquad \mathcal F\subseteq\mathcal E(\mathcal F).
 \label{eq:constraint-erasure-gap}
\end{equation}
The inequality follows directly from feasible-set inclusion. This application-level
functional records a particular lost coupling; it is not a new theorem of
relaxation or interval arithmetic. Its value is that an archaeological qualification
becomes an explicit mathematical constraint whose removal has a calculable effect.

\subsection{A larger regional compilation: the estimand changes the sign}
Gaastra, Lawrence and Tumolo's comparative study supplies a second original
workbook~\cite{gaastra2024}. Its acquired worksheet has 577 rows and 216 distinct
literal site strings. The article reports 245 sites and 786,353 remains; these
are not silently identified with the worksheet's row/string inventory. Repeated
site phases or source assemblages are not independent settlements. The workbook
spans eight country labels and contains 507 \texttt{NISP} and 70 \texttt{TNF}
quantification labels. Number of identified specimens and total fragment counts
are distinct units.

The strict computational cohort admits \texttt{NISP}, exact urban/rural labels
after trimming whitespace, nonnegative integer cattle/Sus/caprine counts, and
exact equality of their sum with the source domesticate denominator. First-failure
exclusions are 70 non-NISP rows, 166 rows with other or unresolved settlement
labels, and one invalid count/denominator row. The remaining 340 rows comprise
172 urban and 168 rural source assemblage rows. These exclusions are scientific
admission choices retained with the original source coordinates; they are not
claims that excluded observations are historically false.

For cohort $g$, let $p_i$ count admitted Sus specimens, $n_i$ the selected-taxon
NISP denominator, and $r_i=p_i/n_i$. Two legitimate summaries are
\begin{equation}
 M_g=\frac{1}{m_g}\sum_{i=1}^{m_g}r_i,
 \qquad P_g=\frac{\sum_i p_i}{\sum_i n_i}.
 \label{eq:mean-pooled}
\end{equation}
$M_g$ weights source assemblage rows equally; $P_g$ weights their identified
fragments. Neither directly estimates an average animal's behavior, and neither
is a causal effect of urbanism. The urban counts are $15{,}867/246{,}979$ and the
rural counts $11{,}070/153{,}688$. The exact-source calculations give
\begin{equation}
 M_A-M_B\simeq+0.006188926487,
 \qquad P_A-P_B=\frac{-147745017}{18978854276}
                 \simeq-0.007784717394.
 \label{eq:regional-reversal}
\end{equation}
The direction reverses without editing any row. This does not refute the original
study's region-, land-use- and time-sensitive interpretation; it identifies the
question encoded by a particular aggregation rule.

Using population covariance over the $m_g$ admitted rows, with
$\bar n_g=m_g^{-1}\sum_i n_i$, straightforward expansion gives
\begin{equation}
 P_g=M_g+\frac{\operatorname{Cov}_g(n,r)}{\bar n_g},\qquad
 P_A-P_B=M_A-M_B+
 \frac{\operatorname{Cov}_A(n,r)}{\bar n_A}
 -\frac{\operatorname{Cov}_B(n,r)}{\bar n_B}.
 \label{eq:covariance-weighting}
\end{equation}
Indeed, $m_g^{-1}\sum_i n_ir_i=\bar n_gM_g+\operatorname{Cov}_g(n,r)$.
Division by $\bar n_g>0$ proves the identity. No model fitting or sampling
assumption is needed. Under the explicitly defined interpolation
$\Delta_\lambda=(1-\lambda)(M_A-M_B)+\lambda(P_A-P_B)$, the zero is
\begin{equation}
 \lambda_*=
 \frac{M_A-M_B}{(M_A-M_B)-(P_A-P_B)}\simeq0.4428999723.
 \label{eq:weight-crossover}
\end{equation}
This sensitivity path mixes two estimands; it is not a historical probability.
Figure~\ref{fig:scope-estimands} makes the dependence visible instead of assigning
a single geographic explanation to the sign.

\begin{figure}[t]
\centering\includegraphics[width=\linewidth]{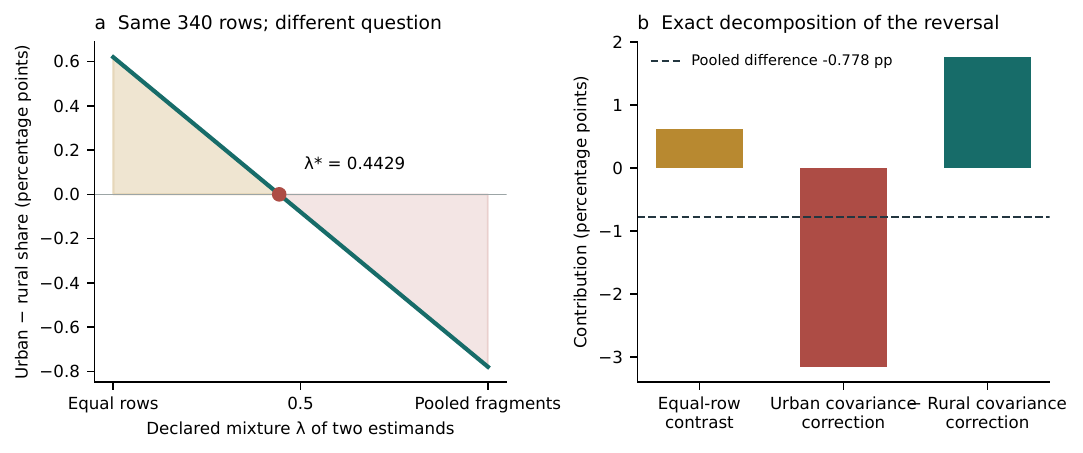}
\caption{Same admitted source rows, different comparison rules. The regional
urban/rural sign differs between equal-source-row and pooled-NISP summaries.
The displayed interpolation is a declared sensitivity path. The right panel
decomposes that reversal into the equal-row contrast and the two cohort
covariance corrections; no confidence interval or causal effect is encoded.}
\label{fig:scope-estimands}
\end{figure}

\subsection{Exact deletion certificates for a declared source-row contrast}
Sensitivity to omitted data has substantial prior art, including the approximate
maximum influence perturbation work of Giordano, Meager and Broderick~\cite{giordano2026}.
For the restricted equal-source-row contrast above, exact deletion has a simple
finite solution. This specialization does not solve general case-removal
sensitivity for fitted models or arbitrary aggregate ratios.

Assume initially $M_A-M_B>0$. Sort the $A$ rates decreasingly as
$x_1\geq\cdots\geq x_m$ and the $B$ rates increasingly as
$y_1\leq\cdots\leq y_n$. Let $X_i=\sum_{k=1}^i x_k$ and
$Y_j=\sum_{k=1}^j y_k$, with $X_0=Y_0=0$. Retain at least one row in each cohort.

\begin{theorem}[Exact sorted-prefix deletion threshold]
For unit-cost deletion of source rows, the minimum number that loses the strict
positive direction is
\begin{equation}
 r_* =\min_{\substack{0\leq i<m,\ 0\leq j<n}}
 \left\{i+j:
 \frac{X_m-X_i}{m-i}-\frac{Y_n-Y_j}{n-j}\leq0\right\},
 \label{eq:prefix-deletion}
\end{equation}
with $r_*=\infty$ if the displayed feasible set is empty. Removing the first
$i$ sorted $A$ rows and first $j$ sorted $B$ rows supplies an attaining witness.
\end{theorem}
\begin{proof}
For fixed deletion counts $i,j$, minimizing the retained $A$ mean requires
maximizing the removed $A$ sum; its $i$ largest rates attain that maximum.
Maximizing the retained $B$ mean requires minimizing the removed $B$ sum;
its $j$ smallest rates attain that minimum. Exchange of any nonextremal selected
rate for an available more extremal rate cannot increase the retained difference.
Thus if any pair of subsets of these sizes loses the direction, the indicated
prefixes do so. Conversely, each displayed prefix pair is an actual allowed
deletion. Taking the minimum over the finite grid proves exactness.
\end{proof}

Prefix sums allow evaluation in $O(mn)$ arithmetic operations after sorting,
using exact rational rates and deterministic ID tie-breaking. Arithmetic bit
cost is additional. The result applies to equal-cost whole-source-row deletion,
not individual fragments, unequal deletion costs or pooled-NISP means. A
negative baseline swaps cohort roles; a zero baseline has no strict direction
to preserve. A finite threshold certifies safety against every permitted
selection deletion of at most $r_*-1$ rows. It does not make such a deletion
acceptable under an unchanged externally fixed cohort.

On the admitted 340-row compilation, the archived exact computation gives
$r_*=3$: every permitted deletion of at most two source rows retains the strict
equal-row direction. One attaining witness removes
\texttt{gaastra-row-0003}, \texttt{gaastra-row-0200} and
\texttt{gaastra-row-0214}. These are source-row coordinates, not three animals
or a measure of historical sampling uncertainty. The certificate concerns a
hypothetically changed selection; those omissions still fail the original
fixed-mask verification request.

\subsection{Binding the scientific question, not only its answer}
Authenticated publication and aggregate verification are established
techniques~\cite{martel2004,liaggregate2010}. The source-specific extension here
makes the scientific query part of the externally checked statement. A scope
query fixes the source snapshot root, exact ordered eligible-ID mask, unit,
question, taxon/phase/recovery policies, estimand, weighting and a canonical
rational interpolation weight. Admitted rows bind original source hashes and
coordinates, cohort, numerator/denominator, qualifications and auxiliary
metadata. Version2 admits both the NISP registry and a distinct Chogha
\texttt{recorded-fragment-field-total} profile, with inference unit
\texttt{source-record-weight}. The same \texttt{certify} and
\texttt{verify\_certificate} entry points dispatch on the typed source profile;
NISP and recorded-field weights cannot be interchanged.

For Chogha Mish the verifier reads the complete original CSV, supplied inline
in a bounded lossless witness or as an independently pinned local source file.
It verifies its byte count/hash, all 16,824 unique URIs, exact ordered headers,
project identifier, exactly-one fragment alias rule, count lexemes and source
phase/date fields. It recomputes all ten context/anatomy aggregate rows before
checking the signed count-bound registry root and eight postcranial openings.
The unresolved predicate-URI correspondence of the three flat count aliases is
retained as a declared source-specific coalescing assumption. A fabricated
aggregate signed with the experiment key fails even if that false aggregate's
root is newly pinned: it is inconsistent with the independently fixed source.

The query fixes both cohorts' complete anatomical support and one of three
modes: pooled recorded fields, fixed shared rational weights, or the whole
shared-simplex envelope. Empty/missing-stratum and cohort-specific vectors fail.
The coded qualification profile fixes known no-sieving, unmerged jaw/tooth
repetition, unestablished physical/animal independence, coevality and domestic
status. Promoting these outputs to MNI, independent animals, unique physical
fragments or same-year deposition fails, including when an attacker pins a
query containing that unsupported promotion. The primary-source interpretation
of these qualifications is reviewed scientific input; cryptography checks its
preservation rather than discovering it from bone rows.

A packet pins the external query and current root, verifies the dedicated-domain
checkpoint, checks unique row openings against the complete ordered mask and
recomputes exact arithmetic. Chogha Mish accepts six valid controls and rejects
27 constructed attacks; its inline-source packet is 1,242,100 bytes. One local
M4 Pro observation measured 0.612s for snapshot construction, 0.624s for
certification and 0.605s for verification. These include full-source parsing.
The separate Gaastra experiment accepts four controls and rejects sixteen
mutations. Proof construction there uses a clarity-oriented path helper;
neither experiment establishes a new performance bound. The Gaastra deletion
radius is not applied to derived anatomical strata, because deleting a stratum
and deleting original specimen records are different interventions.

The typed claim and exact result are authenticated. A caption or plot is covered
only when a trusted rendering path consumes that verified claim and binds the
presentation. Freely edited pixels or captions are outside the row verifier's
guarantee. Neither profile is succinct or zero-knowledge.

Conditional authenticity follows from the earlier tree-binding argument.
An accepted unissued checkpoint gives a signature forgery; altered opened rows
under the admitted root give a hash collision. With neither event, exact mask
comparison excludes omission or transplantation, and deterministic arithmetic
excludes a changed result for the pinned query. For the NISP adapter, admitted-row extraction and unit/mask appropriateness
remain trusted scientific inputs. The Chogha adapter additionally recomputes
extraction from the full pinned CSV, while its source-specific field and
archaeological interpretation profile remains reviewed scientific input. A
textual policy label does not independently prove that its admission rule is
archaeologically appropriate or that all eligible material in the world was
excavated. Current-root freshness needs an external pin or history mechanism.
The public experimental key represents local custody of a calculation, not a
museum's endorsement; this signature domain cannot replace SCENE grants, permits
or atomic spending. The construction's specific contribution is to preserve
source-induced comparison constraints and their sensitivity witnesses across
presentation and verification.

\subsection{Further regions supply different constraints}
The remaining acquired sources enlarge the research programme while retaining
different populations and measurement types. China's 2026 compilation has
3,492 human and 2,919 animal isotope rows~\cite{zhangchina2026}; these are selected
measurements from literature, not NISP abundance. Its human and animal
\texttt{Xingyi} strings carry incompatible coordinates, and some broad phase
fields need an explicit sign/unit convention. An audit produces 88 site-coordinate-phase
join candidates, only 24 with identical coordinates and exact date bounds;
even those remain candidates pending context and specimen-selection review.
The original Xipo workbooks provide 87 selected faunal metadata rows, including
70 Sus, with Yangshao and later Zhou contexts kept separate~\cite{zhangxipo2021}.
A selected isotope-sample fraction cannot be substituted for the site's reported
pig-NISP proportion. These sources motivate authenticated join and chronology
obligations, not a merged cross-country consumption estimate.

The Oman/UAE caprine source has 354 rows in S1, including five without LSI,
and 71 S2 osteometric rows including six modern references~\cite{maini2025}.
A modern Iranian reference is not an archaeological Iranian animal; Daba funerary
and Salut settlement/ceremonial selection are distinct. The two source indices
$X=100B/D$ and $Y=100D/L$ share one lateral depth $D$. For positive admitted
measurement and ratio intervals, joint feasibility requires exactly
\begin{equation}
 [D_-,D_+]\cap
 [100B_-/X_+,100B_+/X_-]\cap
 [L_-Y_-/100,L_+Y_+/100]\ne\varnothing.
 \label{eq:shared-depth-interval}
\end{equation}
Each interval is the projection of one ratio constraint onto $D$; a common
member permits independent choices of $B$ and $L$ satisfying both. Checking
the ratios separately can instead use incompatible depth witnesses.
The exact checker finds all 71 rows feasible under hypothetical closed
$\pm0.05$ mm input boxes and $\pm0.005$ ratio-point boxes. Two rows fail if
input measurements are instead treated as exact. This is a precision-policy
sensitivity result, not a source-error finding: stated caliper accuracy alone
does not prove the hypothetical rounding model. Six focused tests cover these
source and constructed consistency cases.

Lebanon's 2026 author files concern wild Pleistocene fauna roughly 400--100 ka,
not domesticated livestock or later food prohibitions~\cite{russo2026}.
Named-taxon NISP totals 77, 61 and 65 differ from broader totals 559, 63 and 99
that include indeterminate material. A blank Ras el-Kelb Sus count in the main
table and a supplementary 0.0\% entry are retained as distinct source statements.
Tepe Yahya's original isotope workbook has 254 successful caprine records,
30 failed records and 76 analytical-repeat rows from 36 IDs~\cite{seabrook2025}.
Literal quality filtering and source phase-label disagreements change eligible
cohorts; no value is silently repaired. Transjordan's small 2026 comparison is
represented by a clearly labelled eight-column transcription of the original
main table, not an acquired specimen-level supplement~\cite{vermeersch2026}.

The transferable research requirement is therefore concrete: preserve the
eligible cohort, units, mapping decisions, recovery qualifications and shared
variables needed for the particular question. A future corroborating excavation
can strengthen or falsify the narrow Chogha Mish hypothesis, while new source
releases can revise the admitted mappings. The formal checks preserve those
decisions and expose their consequences; the archaeological interpretation
remains answerable to new material evidence.

\section{Independent transfer and archaeological admission}
\label{sec:replication}

The prospective archaeological hypothesis is that \emph{postcranial Sus
representation among recorded cattle, caprine and Sus remains is higher in
floor-use deposits than in occupational refuse, within independently dated
coeval contexts and comparable recovery conditions}. A country label supplies
none of these conditions. Candidate collections and exclusions were recorded
before calculating the independent comparison; after the Chogha tooth-row issue
was identified, the postcranial specification and final source-label mapping
were frozen separately. These are locally timestamped decisions, not an
externally registered preregistration. An exact transfer requires a source-supported
deposition crosswalk, phase and recovery match, and common anatomical support;
an approximate occupation/dump comparison retains a different declared role.

The primary Chogha study materially limits what can enter that hypothesis.
Its p.~446 reports Wheeler's recollection that sieving was not practiced and
explains the omission of detailed spatial analysis. Page~451, footnote~4,
documents repeated spreadsheet entries for teeth belonging to a mandible or
maxilla, which the authors' cleaning merged~\cite{levtov2017}. In the acquired
export, consecutive records CGM12289--12293 specify upper-right P3, P4, M1,
M2 and M3 beside maxilla record CGM12288. Their generic completeness labels do
not establish five independent maxillae. Unique row identifiers cannot establish
fifteen animals, and the seven selected postcranial Sus records have no supplied
individual/refit links. All selected rows share a phase-level context URI and
the editor-assigned 3600--3200~BCE interval: neither identifies coeval rooms or
deaths. The chapter already published the broad temporal pig trend; our focused
inspection is hypothesis generation without a worldwide-priority claim.

An independent English catalogue, OXSTAD16, provides a more limited test.
Broderick's CC~BY~4.0 source has 605 unique catalogue identifiers and 88 fields;
its metadata defines Quantity as the specimens represented \emph{within a
record}, not individuals~\cite{oxstadbones2021,oxstadmetadata2021}. All 606 CSV
lines were retrieved through ordinary public web access after direct local
requests returned HTTP~403. Restoring the original CRLF serialization reproduced
the publisher's legacy MD5 checksum
\texttt{779fab3221ec3063b2d620ab22506e8f}; the resulting original byte sequence
contains 150,707 bytes. This checksum match supports reconstruction, not a
modern collision-resistance assumption or an issuer signature.

Exact floor/refuse replication remains \textbf{unidentified}. The excavation
report distinguishes a real floor sequence from dismantling and cellar
backfill, while candidate floor-sequence contexts 5027/5031 are absent from the
bone catalogue. Context 4014, labelled occupation layer in that catalogue,
belongs to a later latrine; occupation cannot therefore be translated
automatically into floor use~\cite{oxstadreport2019}. The frozen partial test
instead admits literal Phase~2, source-unsieved records, certain selected taxon
labels and explicit postcranial anatomy, comparing declared occupation-layer
labels with dump labels. Question-mark taxa, head/dental material, unknown
anatomy and the original anomalous Phase values 3059/4004 remain excluded
without repair. The 29 admitted records produce
\[
 \widehat\Delta_{\mathrm{partial}}=\frac{0}{9}-\frac{3}{20}=-\frac{3}{20}.
\]
There are two versus eight source context IDs. Forelimb, hindlimb and distal
contrasts are respectively $-1/12$, $-2/7$ and $0$; axial support is absent,
rather than an observed zero. Common simplex weights on the three supported
strata give the exact interval $[-2/7,0]$, with equal-weight index $-31/252$.
The positive direction does not transfer to this declared partial comparison.
It neither confirms nor refutes the stricter floor-use hypothesis, and no
independent-animal significance test or causal interpretation follows.

Larger acquisitions preserve the same admission discipline. Forcello's
published specimen examples are usable, but the required original phase/context
join was described as an Access table on the dissertation CD and was not
acquired~\cite{forcello2015,trentacoste2014}. The Russian collection supplies
123,937 original source-value records from its primary worksheet, with
coordinates, taxonomy, anatomy and contextual qualifications~\cite{rassadnikov2026dataset}.
This is not a harmonized NISP cohort or an individual-animal count.
The bundled compressed exports retain source rows and fields; they are not
original XLSX bytes, all worksheets, or an offline reconstruction of the
workbook. Zvyagino-4 includes floor material, but its yellow-ash comparison has
a roof-material interpretation, and cultural agreement does not establish
simultaneous deposition~\cite{alaeva2025zvyagino}. Roof fill is not admitted
as a refuse proxy merely to obtain a larger comparison.

The receipt therefore includes the complete source and interpretation-policy
SHA-256 pins,\footnote{OXSTAD16 source SHA-256:
\texttt{5d23e911}\allowbreak{}\texttt{ca9697f7}\allowbreak{}\texttt{03bfd7f3}\allowbreak{}\texttt{3db70645}\allowbreak{}\texttt{e6531d45}\allowbreak{}\texttt{fd6e7dc1}\allowbreak{}\texttt{471088cf}\allowbreak{}\texttt{f159ad34}. Frozen interpretation-policy SHA-256:
\texttt{9df7a77d}\allowbreak{}\texttt{3a0762c6}\allowbreak{}\texttt{cfb1114a}\allowbreak{}\texttt{4da22a71}\allowbreak{}\texttt{664a33fc}\allowbreak{}\texttt{e40355d0}\allowbreak{}\texttt{e1c1ab42}\allowbreak{}\texttt{3edd9fb4}.}
original coordinates, a full admission ledger, counting unit, deposition
crosswalk, recovery/phase predicates, supported weights and the exact/partial
transfer role. The executable \texttt{replication\_}\allowbreak{}\texttt{analysis.py} recomputes the
comparison; twelve tests cover source/policy mutation, catalogue and admission
coverage, unchanged source anomalies, context/recovery/taxon/anatomy exclusions,
exact arithmetic, attained weighting bounds and missing or empty support.
These are concrete SCENE admission inputs. The current standalone OXSTAD16
audit verifies local pins and calculations; it does not issue a SCENE scope
certificate. Authentic bytes and reproducible arithmetic leave the archaeological
premises independently open to falsification.

\section{Source-exact reconstruction as an authenticated obligation}
\label{sec:compression}
A submission budget creates a second opportunity for scope loss: a smaller
representation may preserve a rendered shape while changing the actual evidence
file. We therefore implement \emph{source-exact reconstruction profiles} within
SCENE. Each admitted target has a retained role: original publisher member or derived
source-cell extract. Its reconstruction preserves the complete byte sequence, including attributes, order, comments, decimal spelling and whitespace.
The profile separately binds that scientific target, its physical payload and
the packaged decoder files. All manuscript figures remain ordinary PDF inputs;
the TeX compiler need not execute a reconstruction program.

\subsection{What the 2026 methods contribute}
Chen et al.'s Dolphin, published in February 2026, uses the first two XOR bits
and a following seven-zero interval to skip a redundant prefix in oscillating
binary64 data~\cite{dolphin2026}. We implement and measure that \emph{prefix-skip
submechanism} in our own explicitly framed stream, with unchanged-word,
leading/trailing-zero-window and raw-XOR alternatives. This is not a benchmark
of Dolphin's full adaptive framework. Binary64 recovery alone does not recover
original decimal text, so this candidate is admitted only when formatting every
recovered value at the declared precision exactly restores its original token.
Unsupported spellings retain literal bytes.

The May 2026 triangle-removal study distinguishes geometry-equivalent recovery
from original triangle serialization and omits attributes from its experiments
\cite{zalik2026}. That distinction changes our reconstruction target: topology
or image equivalence cannot justify dropping serialization residuals. The
September 2026 Parquet integration of ALP explicitly derives from its 2024
algorithm~\cite{parquetalp2026,alp2024}. We likewise distinguish current experiments
from older component ideas: fixed-point scaling, predictive residuals,
byte/bit planes, dictionary sharing and LZMA are established techniques.
Lossless predictive geometry compression also predates these works
\cite{isenburg2005}. The contribution here is their bounded, source-committed
composition with an executable scientific scope, tested on the actual artifact.

\subsection{An exact lexical profile and a constraint used only as prediction}
The implemented SCENE-LEX3 profile stores a literal/template skeleton for OBJ
records and whitespace, then separate token channels by record kind, arity and
column. Unrecognized records remain literal. Fixed-decimal channels may use
scaled integers, signed deltas, or byte/bit planes; face channels may predict
from earlier components or columns. A finite measured choice includes each
codec's metadata and compressed body. Negative-zero flags, missing slash
components, exceptional spellings and the skeleton are charged payloads.
No vertex, face, material line or texture coordinate is deleted or reordered.

One useful constraint has an entirely different role from archaeological
admission. For shared fixed precision $p$, let $s=10^p$ and let $(x,y,z)$ be the
exact scaled integer lexemes of a supplied normal. With its recorded sign
$\sigma\in\{-1,+1\}$, define
\begin{equation}
 \widetilde z=\sigma\left\lfloor\sqrt{\max(0,s^2-x^2-y^2)}\right\rfloor,
 \qquad e_z=z-\widetilde z,\qquad z=\widetilde z+e_z.
 \label{eq:normal-residual}
\end{equation}
Integer square root makes the predictor deterministic. The residual and lexical
sign, including negative zero, are stored exactly. Therefore nonunit, rounded
or exceptional source normals remain unchanged: unit length is an opportunity
for a shorter residual, not an imposed correction or a claim about measurement
accuracy. The decoder needs earlier $x,y$ channels, the sign bits and every
residual. Erasing any of them destroys the reconstruction obligation even if a
viewer still displays a plausible object. Unsupported precision/range profiles
fall back to an admitted literal representation.

\subsection{Two commitments and a typed revision frontier}
Fix mandatory target files $b_i$ and their original lengths and digests before
compression. An admitted recipe $a$ produces payload $p_{ia}$, metadata $m_{ia}$
and named dependencies $S_{ia}$, subject to
\begin{equation}
 D_a(p_{ia},m_{ia},S_{ia})=b_i,\qquad
 |D_a(\cdot)|=|b_i|,\qquad H(D_a(\cdot))=h_i.
 \label{eq:exact-reconstruction}
\end{equation}
Experiments compare the actual byte arrays as well as the pinned SHA-256 values.
A logical-source root $L$ binds ordered target paths, original/derived roles, lengths and
digests; a domain-separated transport root $T$ binds payloads, codec identifiers
and packaged decoder-file digests. A dedicated-domain checkpoint signs both,
with a fixed decoding policy and bounds. The issuer key and expected roots must
be obtained independently to resist an active substitution attacker. Reading
all pins from the same archive checks internal consistency only. The public
research key does not represent a source institution.

\begin{proposition}[Source-preserving composition]
Assume a trusted deterministic decoder/extractor/verifier runtime, externally
pinned issuer key, source/transport roots, scientific policy $Q$ and current
scope root. An accepted reconstruction followed by accepted full-source scope
verification yields the same qualified result as direct verification on the
pinned original source bytes, except for signature forgery or a source/payload
hash collision. If $Q$ and the logical source members are unchanged between two
accepted storage profiles, the scientific-record frontier is empty; the
transport frontier may still be nonempty.
\end{proposition}
\begin{proof}
Transport checking fixes each complete payload and declared decoder dependency.
Bounded deterministic decoding followed by the original length/hash check fixes
its source output, unless the pinned digest is collided. Full-source extraction
therefore receives the same bytes and policy as the direct calculation and
recomputes the same registry, selection and exact result. A checkpoint outside
the admitted issuance history would instead forge its signature. Applying this
argument to both profiles gives equal scientific records. Their payload and
decoder-file records need not agree, so their transport frontiers need not be
empty. Changing $Q$ invalidates the stated equality premise even if every source
byte is unchanged.
\end{proof}

This connects compression to revision-cut reuse without confusing two events.
An unrelated storage change can preserve a scientific cut after reconstruction,
while a changed taxon mapping or weighting is a scientific-policy revision.
Transport freshness must still be verified against its own current pin. The
implementation fixes decoder identifiers and never executes downloaded recipes,
shell instructions or serialized objects. It limits individual source outputs
to 128MB, the total declared reconstruction to 256MB and each XZ decoder's working
memory to 64MiB. The last bound is not a total process-memory bound. File checks
cover the packaged Python dependencies under a trusted bootstrap, interpreter,
liblzma and cryptography runtime; code already executing cannot authenticate its
own replacement. Path traversal, symlinks, malformed frames, trailing compressed
data and changed existing source destinations are rejected.

\subsection{Storage cost is a union, authorization cost counts occurrences}
BagIt and OCFL already separate inventory/fixity from logical organization, and
conditional reconstruction is longstanding delta-compression practice
\cite{bagit2018,ocfl2024,vcdiff2002}. Let a finite acyclic reconstruction graph have
nonnegative stored-node costs $c_v$, mandatory dependency closure $\cl$, and
selected recipes $A$. Its additive physical cost is
\begin{equation}
 F(A)=\sum_{v\in\cl(A)}c_v.
 \label{eq:storage-union}
\end{equation}
Shared dictionaries and decoders are counted once in this union. A fixed
closure model is ordinary monotone submodular weighted coverage: the additional
nodes introduced by one recipe can only decrease as $A$ grows. For binary recipe
selectors $x_{ia}$ and stored-node indicators $z_v$, a finite joint model is
\begin{equation}
\begin{aligned}
 \min_{x,z}\quad&\sum_v c_vz_v+\sum_{i,a}q_{ia}x_{ia},\\
 \text{subject to}\quad&\sum_a x_{ia}=1 &&\text{for each target }i,\\
 &z_v\ge x_{ia} &&\text{for all }i,a,\ v\in\cl(i,a),\\
 &x_{ia},z_v\in\{0,1\}.&&
\end{aligned}
 \label{eq:recipe-selection}
\end{equation}
Here $q_{ia}$ includes selector/recipe serialization. The model does not include
layout-dependent solid compression unless those costs are explicitly added.
The implemented profile selects measured per-target candidates with one common
fixed decoder, rather than solving this general joint optimization. The complete
outer ZIP byte count is the authoritative size test.

SCENE endorsement budgets instead charge each disclosed occurrence. Reusing
one compressed source in many displays can save storage while consuming many
authorizations. Identifying the two costs would allow compression or deduplication
to manufacture endorsement capacity. Keeping their domains separate is an
operative security condition, not a cryptographic property of the compressor.

\subsection{Measured source recovery and archive budget}
The benchmark includes all thirteen OBJ members: the SEAMS mummy and twelve
\c{C}atalh\"oy\"uk models. It compares raw-member DEFLATE9, XZ6 and complete
SCENE-LEX3 frames; the frame includes its skeleton, sign/exception streams and
codec metadata. Every selected payload restores its exact original source bytes.
Across 129,303,415 original bytes, raw DEFLATE9 uses 32,988,098 bytes, XZ6
uses 21,795,388 and the complete thirteen SCENE-LEX3 frames use 16,402,504:
24.7\% less than XZ6. The release selects XZ6 for one model where it is smaller;
the thirteen selected model payloads total 16,362,028 bytes. A single local run
measured total encoder times 11.908s, 52.446s and 191.651s, respectively, and decoder
times 0.094s, 1.001s and 20.708s. Encoder selection effort and lexical checking are
included for SCENE-LEX3; these are not isolated throughput or optimality claims.
Cumulative process peak RSS was 760,004,608 bytes. Figure~\ref{fig:source-compression}
shows each actual member rather than only the aggregate saving.

The 2026-inspired candidate was evaluated on 80 numeric channels; 69 exercised
its prefix-skip branch, for 395,697 such transitions. It won only one constant
channel, where the skip branch was unused. Thus this corpus supplies a measured
negative transfer result for that mechanism: the observed gain comes from
other exact lexical/integer/residual choices. A comparison that called the
selected compressor ``Dolphin'' would misidentify the implemented method.

\begin{figure}[t]
\centering\includegraphics[width=\linewidth]{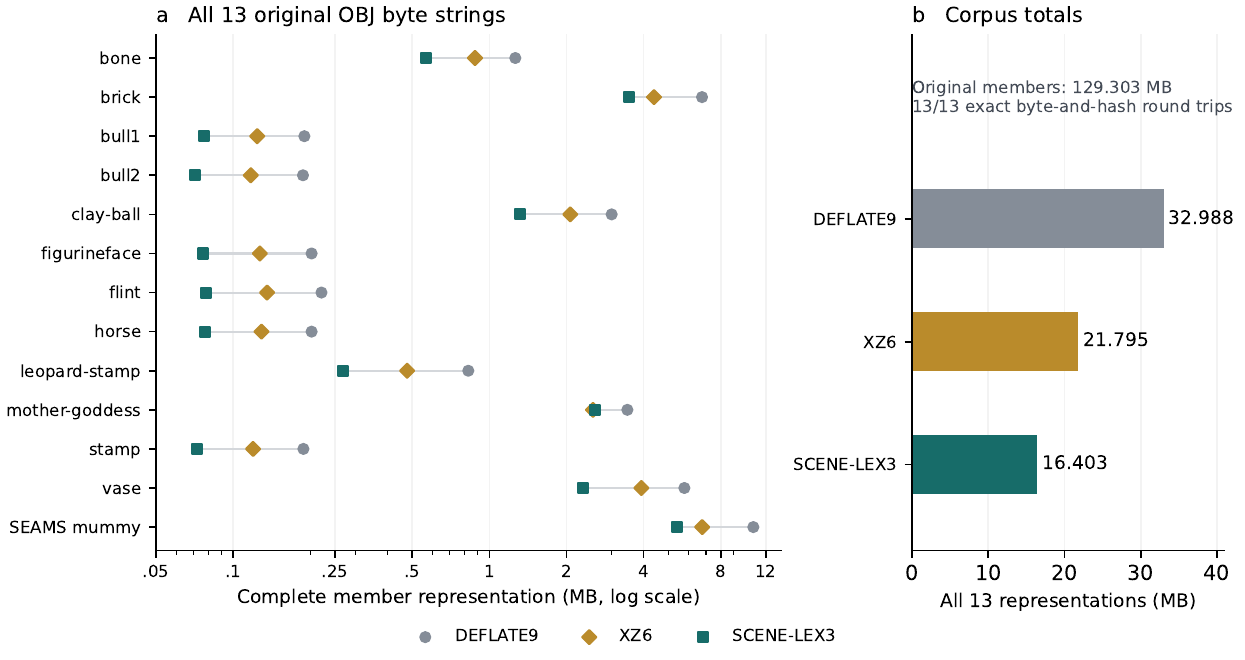}
\caption{Complete exact-byte representations of all thirteen source OBJ members.
The left axis is logarithmic; the right panel gives additive frame totals.
DEFLATE9 counts a raw RFC1951 payload, XZ6 its full frame, and SCENE-LEX3 its
header, skeleton and compressed channels. Every method restores all original
bytes. Decoder/code/recipe costs, other source assets and outer-ZIP structure
are measured separately at package level.}
\label{fig:source-compression}
\end{figure}

The original photographs, textures, workbooks and unchanged licensed ontology
remain available. Selected large CSVs use exact-byte XZ storage, with their
original line endings and count lexemes restored before source parsing.

The release must satisfy two physical budgets. Let $U$ be the sum of the
stored files' lengths after the outer ZIP is opened, $Z$ the complete ZIP length,
and $r$ a declared allowance for the processed PDF and build products. We require
\begin{equation}
 Z<M,\qquad U+r<M,\qquad M=50{,}000{,}000\ \text{bytes},
 \quad r=2{,}000{,}000\ \text{bytes}.
 \label{eq:two-upload-budgets}
\end{equation}
The earlier $49{,}778{,}451$-byte ZIP had $56{,}697{,}985$ bytes of opened members:
checking $Z$ alone was insufficient. The corrected packager measures both actual
container properties and fails at either boundary. The allowance is a conservative
local engineering margin, not an assertion of arXiv's exact server accounting.
The much larger reader-restored logical source volume is a third quantity and
is reported separately; TeX processing does not run that restoration.

The two costs can also be attached componentwise to a fixed dependency closure:
$\mathbf F(A)=(\sum_{v\in\cl(A)}u_v,\sum_{v\in\cl(A)}z_v)$, with recipe and container
overheads additionally counted. Shared dependencies remain unions in each
coordinate. A plan feasible for one resource can fail the other; direct complete
package measurement remains necessary for layout-dependent terms. This refines
Equation~\eqref{eq:recipe-selection} to a constrained multi-resource problem
without changing any scientific source value or comparison policy.

To reduce $U$ as well as $Z$, the new profile uses exact-byte XZ for eligible
data records and libjxl 0.11.2 JPEG reconstruction for all 77 source JPEGs
\cite{libjxl0112}. They total 22,725,571 original bytes and 18,753,153 stored JPEG XL
bytes; all 77 original JPEG byte arrays and SHA-256 values are recovered. This is
JPEG-bitstream reconstruction, not pixel decoding followed by a new JPEG encode.
The wrapper requires reconstruction metadata, fixed decoder version/options,
a 60-second timeout, an output-file limit and the original source digest.
Native decoder memory is not bounded by this wrapper. The trusted installed
libjxl runtime is an additional dependency, while no native binary is hidden in
the archive. These established codec capabilities are attributed, not claimed
as a 2026 invention.

The outer ZIP retains all manuscript figure PDFs. Original photographs,
textures, workbooks and source metadata restore to their original paths with
unchanged bytes and roles. Generated result arrays, duplicate plots, browser
captures and HTML are rebuilt by the offline runner. Earlier investigator OBJ
gzip wrappers were replaced while publisher OBJ bytes/digests were preserved.
The Russian exports retain their source-value-extract role; compression does
not make them an original-workbook archive. The artifact records the final
$Z,U,r$ values and verifies every restored target before scientific processing.

\section{Authenticated query contracts across execution and materialization}
\label{sec:database}
The database problem is to preserve a \emph{qualified relational question} across
storage reconstruction, alternative execution plans and cached results.
Authenticating aggregate answers has substantial prior art
\cite{martel2004,liaggregate2010}. Our database contribution is an executable
source-specific query contract that connects complete source extraction,
relational bag semantics, anatomical support and the existing SCENE scope
verifier. The independent SQL execution is checked against that verifier,
rather than treating a database's declared source hash as proof of its contents.

A contract $\Theta$ binds the source-byte and complete parsed-record digests,
current scope root, original-column mapping, registered parameterized statement
hashes, parameters, count unit, result schema/order and scientific qualifications.
Its eligible source bag $B_\Theta(R)$ retains each record's URI, ordinal,
source-value hash and recorded weight. Unique URIs prevent an accidental duplicate
import; they do not deduplicate animals, teeth or physical fragments. For the
complete Chogha source, SQL \texttt{COUNT(*)}=16,824 whereas
\texttt{SUM(recorded\_weight)}=16,825. Substituting one for the other changes
the counted unit. Blank original fields remain empty strings; the importer
validates integer-valued decimal lexemes and exactly one active count alias
before SQLite casts the coalesced value. SQLite's own count/SUM/NULL semantics
are not a substitute for this source admission check~\cite{sqliteagg2026}.

For a fixed grouping operator $G_\Theta$, an answer carries both grouped counts
and complete selected-row coverage. The acceptance condition is
\begin{equation}
 \operatorname{Accept}_\Theta(A)
 \iff V_{\rm SCENE}(A_{\rm scope};R,\Theta)=1
 \ \land\ B_A=B_\Theta(R)
 \ \land\ G_A=G_\Theta(B_\Theta(R)).
 \label{eq:SQL-contract}
\end{equation}
Here the existing full-source scope certificate establishes the selected bag and
expected aggregates from the independently pinned original CSV. A received
coverage hash alone does not establish completeness: the ordered URI/ordinal,
record-hash, cohort, anatomy, taxon and weight tuples must equal that verified
bag. Missing groups are rejected rather than interpreted as observed zeros.
Only fixed, registered statement templates execute; packet-supplied SQL text
cannot enter the database engine.

This condition separates result equivalence from contract equivalence.
Alternative scan, index and materialized-view plans may answer the same contract.
Changing numerator parameters, anatomical support, unit or source version cannot
answer it merely because a scalar ratio happens to agree. Even multiplying a
numerator and denominator by the same factor fails the exact-group check. If two
storage recipes restore identical source bytes, the source and query roots, bag
and qualified result are unchanged. A different transport payload alone therefore
does not stale this query contract. A changed source root or parameter contract
invalidates the materialized-view binding; these are different revision domains.
The argument composes deterministic relational evaluation with SCENE's existing
source-binding result, under its stated runtime/key/root assumptions. It does
not claim a new SQLite index or a succinct proof that a physical plan executed.

\paragraph{Actual relational workload.}
The artifact imports all 16,824 records and all 59 original field lexemes into a
fresh SQLite database, with an explicit source-column dictionary, URI uniqueness
and count-alias consistency constraints. Complete base-value verification
recomputes every parsed record hash against the original source, including fields
unused by the current aggregate. The declared postcranial query fixes
\texttt{Protoliterate}, both context labels, the five exact taxon labels and four
anatomical strata. It returns eight cohort/stratum groups covering 225 source
records, 5/61 versus 2/164, and the same exact contrast $349/5002$ through the
\texttt{scope\_contrasts.certify} and \texttt{verify\_certificate} interface.

Three measured physical strategies use a table scan, a covering index, and a
covered materialized bag. Their \texttt{EXPLAIN QUERY PLAN} outputs are retained
as diagnostics, never parsed into an authentication decision; SQLite warns that
this format can change between versions~\cite{sqliteeqp2026}. The experiment
separately reports import/index/materialization cost, database bytes, aggregate
SQL time, full source-value checking and full-source answer verification.
In one SQLite 3.53.3 run, aggregate SQL took 4.215, 0.127 and 0.079ms,
respectively; database files occupied 24,399,872, 27,901,952 and 27,975,680 bytes.
Import and complete base-value verification took 0.753 and 0.910s, while
full-source answer verification added 0.570, 0.572 and 0.575s. These costs are
reported separately so the small aggregate time cannot conceal verification.
The materialized view binds its coverage and scope to the complete contract
before its answer can be accepted. Eight constructed attacks cover changed
aggregates, omitted groups, rehashed incomplete coverage, duplicate coverage,
changed parameters, unregistered statements, animal-unit promotion and stale
view bindings; all are rejected. Twenty tests additionally exercise relational
keys, count-alias constraints, unchanged-ratio tampering, source/column mutation
and storage-wrapper invariance. Timings are one local sequential run, not a
claim of optimizer novelty or end-to-end speedup: full-source verification is
additional to the aggregate SQL operation. The generated database is an
inspectable reproduction output; it does not replace the pinned original CSV.

\section{A frozen documentary-language extraction and revision pilot}
\label{sec:nlp}
The language question is whether a source-bound parser preserves the distinction
between a representation claim, a qualification and a calendar statement before
they become scene/database records. We implement a restricted English
finite-state/rule parser, with no trained model or language-model inference.
Finite-state information extraction, lexical negation handling and rule-based
temporal normalization have substantial precedents~\cite{fastus1993,negex2001,heideltime2010}.
We evaluate a particular documentary-frame contract and its revision boundary;
we do not identify ordinary regular expressions as a new parsing algorithm.

\paragraph{Natural source records and development separation.}
The corpus includes every description sentence from the 33 Cleveland Museum of
Art records in the visual-source inventory, all 15 Metropolitan Museum titles,
and all 48 literal catalogue-date strings: 243 records comprising 180 prose
sentences, 15 title fragments and 48 date phrases. Text remains an exact substring
of a licensed CC0 JSON field, with its original source digest, Unicode-codepoint
offsets and text digest. No sentence was selected because a trigger or parser
succeeded. Two near-duplicate object pairs were merged before splitting, using
word-trigram Jaccard similarity of at least 0.65 on segments of at least 12 words.
The 48 objects therefore form 46 source groups. A fixed SHA-256 group rule assigns
188 records in 38 groups to development and 55 in 8 groups to test. Whole objects
and duplicate clusters cannot cross the split; common short date/title strings
are not evidence of independent linguistic constructions.

Reference frames were prepared in an investigator's AI-assisted workflow, blind
to the evaluated parser's predictions. They are provisional reference annotations,
not independently human-adjudicated gold or an inter-annotator-agreement study.
Structured catalogue dates remain separate weak metadata and are never substituted
for text labels. After three development inspections, the parser's code/lexicon
hash was frozen before test or external-corpus evaluation. A blind source audit
then corrected two development-only cue labels; test bytes did not change.
Current development scores use that corrected reference, rather than claiming
perfect fit to the earlier version. The artifact retains both audit receipts.

\paragraph{Frame grammar and measured accuracy.}
Admitted relations have an inflected \emph{depict}, \emph{show} or \emph{represent}
trigger, a bounded concrete argument head and source-scoped polarity/modality.
Entity classes include documentary king/queen/man/woman/deity/animal/other labels;
they do not assert biological sex, ancient authority or a real person's identity.
Unsupported paraphrases abstain. Dates require an explicit textual era and
normalize inclusive BCE-negative/CE-positive intervals without year zero.
For example, an explicit third century AD becomes 201--300; an unmarked century,
relative age or dynasty number is not silently converted using catalogue metadata.
Source spans, date alternatives and record genre remain explicit.

Exact-frame micro scores use $P=TP/(TP+FP)$, $R=TP/(TP+FN)$ and
$F_1=2TP/(2TP+FP+FN)$; undefined empty-target ratios are reported as null.
An abstention on an admitted reference frame counts as a false negative.
On the 55 held-out records, the frozen parser matches the single positive
relation frame and all ten admitted date frames. These tiny positive denominators
cannot establish broad relation/date generalization. The natural relation
positives do not test negative or uncertain relation scoping. Negation has
one matched cue. Uncertainty extraction instead gives 9 true positives, 4 false
positives and 2 false negatives: precision $9/13$, recall $9/11$ and $F_1=3/4$.
Whole-record outputs match 51/55 records; grouped results are also retained.
Two near-duplicate records falsely treat \emph{about} in \emph{brought about}
as epistemic uncertainty, and two additional errors concern cue boundaries.
No rule was changed after observing these failures.

Two declared ablations distinguish the target from word presence. A deliberately
weak, unlinked entity-mention heuristic produces 49 spurious held-out relation
frames under the strict trigger/argument contract; it is not a competitive
information-extraction baseline. Removing qualifications from the same grammar
loses all admitted negation/uncertainty cues. Its positive relation score remains
unchanged in this sparse test because the positive relation has no such
qualification. The artifact reports all slot scores, genres, empty records,
coverage and failures instead of equating document-level empty answers with
correct extracted entities.

\paragraph{Fresh source receipts and semantic-cut reuse.}
Let $f_\theta(t)$ be the extracted span frames and $s_\kappa(f_\theta(t),t)$
the declared slot compiler. A dedicated-domain receipt binds
\begin{equation}
 R=\operatorname{Sign}\bigl(H(t),H(\theta),H(\kappa),
             f_\theta(t),H(s_\kappa(f_\theta(t),t)),\text{record},\text{genre}\bigr).
 \label{eq:nlp-receipt}
\end{equation}
Verification pins the current text, record, genre and rule identity, recomputes
the extraction and checks the separately scoped slot contract. It authenticates
this parser's documentary output, not linguistic correctness or historical truth.
A semantic decision cut may avoid a changed raw-text leaf when its admitted
slot record is unchanged, but a fresh source-to-frame receipt remains mandatory.
An old receipt cannot answer a new text merely because its normalized slots
happen to agree. Source/slot-policy changes remain distinct revision events.

The predeclared fourteen development-derived edit pairs are explicitly
investigator constructions, not publisher revisions. All fourteen produce the
intended extraction distinction; four unchanged-slot cases transport their cuts,
and every old source receipt is rejected against the revised text. An expanded
systematic audit has 358 constructed variants: 243 whitespace changes, 57 cue
removals, 44 date substitutions, 7 negation insertions and 7 entity substitutions.
It matches 353 intended slot outcomes and rejects all stale source receipts.
The five failures are conservative false invalidations: internal whitespace in
multiword cue lexemes remains in the frozen slot representation. Those benign
edits unnecessarily prevent reuse; they do not admit stale text. The fixed
pipeline and these failures are retained for inspection. Twenty-eight software
tests check source grounding, scoping, chronology, abstention accounting,
external receipt conditions and actual revision-cut verification. The frozen
lexicon's age-role collapse is another disclosed out-of-corpus counterexample;
it is not validated as an inference from a child's representation to an adult
social role. This pilot motivates broader independent language evaluation,
without returning to a learned model or inferring source consent.

\subsection{External human-annotated temporal evaluation}
\label{sec:external-temporal}
To distinguish investigator-curated museum references from an external human
benchmark, we applied the unchanged date grammar to every English document in
the acquired AncientTimes and WikiWars releases. AncientTimes was manually
annotated for historical temporal expressions; WikiWars provides human TIMEX2
annotations in historical narratives
\cite{strotgen2014ancienttimes,mazur2010wikiwars}.
The parser rule hash was frozen at 20:38:15 UTC on 27 October 2026, before
external annotation inspection; neither corpus was used to adjust the rules.
The evaluation adapter was written subsequently. These are transfer scores for
the complete acquired English editions, not a reproduction of either paper's
experimental protocol. Two AncientTimes articles derive from WikiWars, so the
corpora do not constitute independent replications.

First, we supplied each human TIMEX span to the grammar and measured
\emph{year projection}. A complete four-digit calendar year in a human DATE,
TIME, or untyped TIMEX2 value, optionally prefixed by BC and followed by
calendar/time refinements, projects to $(y,y)$; BCE is negative, CE positive,
and year zero is rejected. Durations, sets, references, incomplete years,
centuries and decades remain counted but unsupported. Month, day, time and
season correctness are outside this projection. Exactly one returned frame
constitutes an attempt, even when it covers a proper substring of the supplied
span; zero or multiple frames abstain, without selecting a candidate using the
gold value. Covered precision is correct/attempted and coverage is
attempted/eligible. Human values may incorporate surrounding discourse, whereas
this task supplies only the span text. Consequently these scores cannot be
interpreted as complete TIMEX normalization or span-extraction accuracy.

AncientTimes contains five documents and 306 body TIMEX annotations in the
pinned release. Of 210 complete-year values, the grammar attempts 124 and
matches 123: covered precision $99.19\%$, coverage $59.05\%$, and correct
year-pair recall $58.57\%$. The remaining 96 annotations have unsupported
values/types. WikiWars contains 22 documents and 2,659 body annotations; of
2,088 complete-year values, it attempts and matches 83: covered precision
$100\%$ and coverage/recall $3.98\%$. Its 571 other annotations are unsupported.
Relative/context-resolved dates and ordinary years without the explicit era
required by the museum grammar account for substantial abstention. The one
AncientTimes mismatch is a season encoded under one gold year whose surface
expression spans two BCE years; it remains an error under this scoring rule.
Correct year pairs constitute only $40.20\%$ and $3.12\%$ of all human spans,
respectively. The high covered precision therefore does not establish broad
temporal understanding.

Second, we ran the same grammar on each complete document and scored exact
source spans against \emph{all} human body temporal annotations, including
unsupported types. AncientTimes yields 124 predicted spans, 78 exact matches,
precision $62.90\%$, recall $25.49\%$, and F1 $0.36279$. WikiWars yields 103
predictions and 49 matches, precision $47.57\%$, recall $1.84\%$, and F1
$0.03548$. These separate extraction scores expose the restriction's cost;
they include neither normalization credit nor an asserted comparison with
HeidelTime. The cryptographic contract binds the source and declared parser
output even when linguistic recall is poor; authentication supplies no missing
linguistic inference.

All spans address XML-decoded TEXT Unicode codepoints; nested elements and
tails retain their order, with no whitespace normalization. DCT/DATETIME
outside TEXT is excluded. Reconstructed offsets agree with all 306 supplied
AncientTimes offsets and all 2,640 supplied WikiWars offsets; 19 WikiWars
annotations lack that attribute. The ancillary evaluator rejects changed
archive or parser hashes and records original member and decoded-text hashes,
aggregate denominators, and unsupported types. The source ZIPs contain no
corpus license and explicit original redistribution permission was unresolved;
we therefore distribute code, acquisition checksums and aggregate receipts,
with separate acquisition required to rerun this external study. The museum
experiment remains available offline in the package.

\section{Dimensional regression fixture}
\label{sec:archaeology}
Petrie's digitized 1883 publication reports the Khufu casing-base mean as
$9068.8$ inches and reconstructed original height as $5776$ inches
(chapter VI, sections 21 and 25)~\cite{petrie1883}. Reproducible conversion using
$25.4$ mm per inch gives $230.34752$ m and $146.7104$ m. The signed fixture rounds
these to $230348$ and $146710$ mm. This convention does not calibrate an 1883
survey to modern SI. Petrie's probable errors, $0.5$ and $7$ inches, are preserved
as historical descriptions, not recast as standard deviations or guaranteed bounds.

The El-Kurru field report gives Ku. 1's plinth side as $26.65$ m and reconstructs
original height at approximately $34.5$ m~\cite{emberling2015}, p.~60. No numerical
uncertainty is supplied, so none is invented. Reported upper-course block-cut
angles of $70$--$75$ degrees describe a different scope from the ideal-envelope
angle $\arctan(2\times34.5/26.65)\approx68.88$ degrees. Courses step back. The
artifact retains that qualifier; transferring the block-angle claim to the whole
envelope is a benchmarked scope-substitution attack.

The two newly generated meshes each have five vertices and six triangular
faces. They are purposively selected fixtures, not a representative cultural
sample. A supplemental Nuri base measurement has no reported height in the
selected source and produces no mesh. We inspected digitized 3D resources from
Digital Giza and the Zamani Project; the examined Meroe model was not offered as
downloadable with a redistribution licence. No viewer geometry was extracted.
This small pyramid fixture therefore uses numerical transcriptions from digitized
publications rather than those scan datasets. The principal actual-image and
photogrammetric inputs are described in Section~\ref{sec:realdata}.
Source locators, rights notes, file hashes and retrieval
limitations accompany the cards. No copyrighted publication pages or photographs
are redistributed. This fixture's newly transcribed facts and generated envelopes
are CC0; software is MIT licensed. The SEAMS input and its derivatives retain
their separate CC BY 4.0 attribution.

This distinction follows the London Charter's emphasis on sources, documentation
and sustainable visualization~\cite{london2009}. Geometry makes evidence obligations
concrete and reconstructible; it supplies neither a universal theory of ideology
nor secret key entropy. For a deterministic key $K=f(G)$ derived solely from public
geometry $G$, $H(K\mid G)=0$ in Shannon's sense~\cite{shannon1948}. This observation
does not rule out computational hardness of other public-instance problems; it
rules out treating this publicly specified deterministic key as a secret.

\section{Reproducible computational evaluation}
\label{sec:evaluation}
All reported runs used an Apple M4 Pro, 48 GiB memory, macOS 26.3 arm64, Python
3.14.2, cryptography 46.0.4, NumPy 2.4.1, SciPy 1.17.0, Matplotlib 3.10.8 and
Pillow 12.1.0.
The fixed protocol clock is a deterministic test input, not an observation of
the execution date. No model was trained or queried by the method.

\paragraph{Capacity instances.}
The artifact evaluates 64 fixed-closure budget pairs for the two-monument scene,
eight shared-prerequisite stars, the Fano obstruction, and 60 seeded small
hypergraphs. A general-purpose integer solver and independent memoized exhaustive
enumeration agree on all 133 instances. Fractional capacities are computed by
linear programming with independently solved price duals; recovered rational
prices certify their reported upper bounds. The Fano instance has integer
capacity one and fractional bound $7/3$, the only strict floor gap in this suite.
The eight-alternative star has true capacity one and erased capacity nine.
These small examples verify implementation and deductions; they do not establish
large-instance solver complexity or historical causal regularities.

\paragraph{Adversarial protocol suite.}
Ten mutation families have 40 variants each: purpose, audience, narrative,
arrangement description, height, spatial position, uncertainty, dependency edges,
block/envelope scope and cloning. Replay, fresh-permit overdraft and transitive
source omission add three cases, for 403 unauthorized cases and one positive
control. The source-omission case is a deliberately synthetic cross-source
dependency, not an archaeological hypothesis. The verifier accepting only authenticated input cards admits all 403 attacks;
exact scene binding with closure but without spending admits two; using only
direct sources admits three. Full closure and atomic spending admit none.
All four accept the single positive control. The ablations are explicitly implemented toy verifiers,
not evaluations of C2PA or other deployed products. Variant counts indicate
constructed coverage, not independent population samples or empirical attack probabilities.

\paragraph{Independent tests and interface checks.}
The automated tests include every Merkle index for sizes 1--65, 960
selected-object views over all 64 four-node DAGs, mandatory-opening omissions,
concurrent replay and overspending, atomic rollback, malformed encodings,
key/ledger substitution, scope, expiry, source-byte modification and corrupted
audit state, plus real-source image proofs and mesh-buffer extraction. Capacity
tests independently check the dual, Hall boundary, integrality and comparative
statics. The final verification report records the complete test count.
In actual Chromium, Python and browser
fingerprints agreed; camera changes preserved the manifest, cloning increased
occurrence cost and marked the scene as an unendorsed draft, and restoration
returned the original fingerprint. The viewer checks no institutional endorsement.
Its JavaScript fingerprint is only an optional content-consistency aid.

\paragraph{Local timings.}
Eleven within-process repeats per size, in seeded shuffled order after warmup,
measure roots and single-object disclosures for 2, 8, 32, 128 and 512 occurrences.
Median root time increases from $0.483$ to $14.065$ ms; disclosure verification
from $0.284$ to $0.366$ ms. The single-object proof grows from $4649$ to $7129$
bytes, including the mandatory source/evidence closure. A two-source committed
debit takes a median $2.409$ ms. Figure~\ref{fig:timing} shows observed ranges;
CSV rows retain every measurement. Proof construction recomputes trees in this
clarity-oriented implementation and is not an optimized production benchmark.
The timings include validation and serialization, are susceptible to shared-host
load, and should not be interpreted as calibrated cross-machine performance claims.

\begin{figure}[t]
\centering\includegraphics[width=\linewidth]{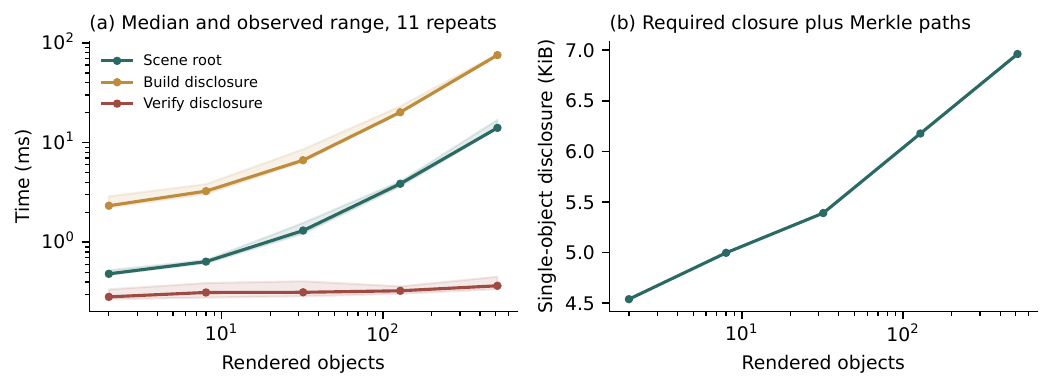}
\caption{Local artifact timings on the stated M4 Pro. Shading shows observed
minima and maxima, not confidence intervals. Proof size includes semantic closure,
so a constant number of selected objects need not imply constant-size disclosures.}
\label{fig:timing}
\end{figure}

\subsection{Revision transport at source-record scale}
The revision experiment constructs a universe of 1,168,790 uniquely addressed
source records: 3,374 original G\"obekli cells, 191,154 nonempty cells in the
original Cassidy workbook, and 974,262 fan-triangles from the twelve original
\c{C}atalh\"oy\"uk OBJ files. A triangle digest binds the original file hash,
fan index and three original vertex indices. The file digest transitively binds
the referenced coordinate bytes; this benchmark does not open each triangle's
coordinates or replace the mesh extraction verifier. We deterministically sample
10,000, 100,000 and 1,000,000 distinct records from this universe, without
replicating rows to inflate scale. These are storage/proof records, not independent
people, artifacts or archaeological observations.

Each size adds 64 explicit local review records and a fixed 256-gate policy:
two four-record conjunctions, their disjunction, and one registered-defeater
eligibility guard per query. The query meaning is presence and digest of the
admitted archived records under declared qualifications. Distinct objects are
not thereby asserted to be historically interchangeable. Query count and rule
size remain fixed while the authenticated registry grows, so this isolates
registry-scale update and proof costs; it does not benchmark inference over
one million independent historical claims.

Three separately signed revision branches change local qualifications or review
inputs while preserving every original source file. First, 100/1,000/1,000
records outside the selected frontiers receive a qualification edit. Second,
one selected-branch qualification changes for each of 16 queries, without changing
the truth bits. Third, eight registered review defeaters become active. At all
three sizes, the first case transports 64 original cuts; the second transports
48 original cuts and recovers the other 16 through old alternative cuts; the
third transports 56 and has no avoiding positive cut for the remaining eight.
Fresh current cuts independently verify the corresponding status for all 64
queries. Deleting a real changed opening from each of the nine deltas is rejected.
These are controlled mutation experiments, not observations of source authors
revising their publications.

At one million records, a full initial registry takes 20.618 s in the archived
run. The three persistent updates take 19.603, 0.722 and 0.428 ms, rehashing
2,012, 256 and 32 nodes respectively. Their complete serialized deltas contain
632,290, 31,503 and 7,696 bytes; batched verification and transport of the 64 cuts
takes 112.156, 127.708 and 127.534 ms. Creation of 64 fresh cuts from the already
materialized current state takes 2.740, 2.728 and 2.401 ms. Thus transport is
valuable as a verifiable relation between checkpoints; these observations do
not show it to be faster than asking a trusted current-state holder for fresh
cuts. Update, delta generation, public verification and fresh-cut construction
are separately timed operations. Changed IDs in the first case are selected
lexicographically outside the frontiers and can share tree paths. A single
measurement per size/branch supports functional scaling inspection, not a
confidence interval, random-update average or universal speedup.

\begin{figure}[t]
\centering\includegraphics[width=\linewidth]{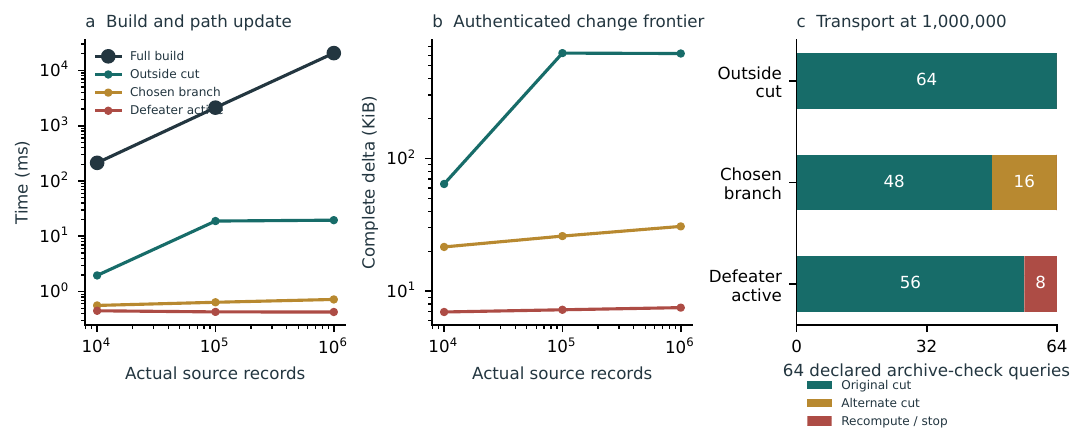}
\caption{Actual source-record registry scaling and three controlled revision
branches. Full build and incremental update are different operations. Delta
bytes include all changed leaves and unchanged subtree frontier hashes; query
count stays 64. At one million records, alternative cuts recover 16 touched
certificates, while eight active defeaters prevent positive eligibility reuse.}
\label{fig:revision-scaling}
\end{figure}

\section{A continuing research program and its limits}
The immediate archaeological priority is an independent test of the Chogha Mish
context hypothesis with matched recovery, dating and deposition. A second
prospective hypothesis asks whether local pig provisioning can lag a crop/diet
transition; unresolved Chinese site/phase matches and local ecological baselines
currently prevent its test. These questions are not inferred from a verifier's
rejection rate. Qualified new observations can revise source admission and the
comparison policy without overwriting the earlier evidence.

The mathematical model opens several concrete questions: minimum-cost avoiding
cuts under shared subgraphs, exact or approximate transversal computation for
cut families, compact representation
of exponentially many witness alternatives; exact algorithms for restricted
incidence families; stronger integral certificates; and scene configurations
with reconstruction-school compatibility. Context can alter admissibility and
budgets, but a learned or discretionary context classifier would introduce a
new trusted boundary requiring separate evaluation. Revocation, key rotation,
multi-ledger consistency and independent checkpoint witnesses are also future
work rather than properties of the local prototype.

The human hypothesis remains untested. The accompanying proposed protocol uses
a $2\times2\times2$ design varying contextual framing, qualifications and an
enforced export checker. Its primary belief outcome is erroneous attribution of
institutional endorsement; successful export is a separate computational endpoint.
Neither is inferred from rejecting a forged packet. The plan includes source
recall, randomized stimuli, mixed-effects logistic analysis and pilot-based power
simulation, without fabricating a sample size or participant results. Follow-ups
at days 1, 7, 30 and 90 occur after immediate debriefing and therefore measure
retention after correction, not uncorrected natural persistence. No participants
have been recruited and no ethics approval is claimed.

Source cards, witness policies, renderer versions, signing identities and
preregistered analysis decisions should form separate versioned registries.
Archaeological and community expertise is needed to review qualifications and
interpretive equivalence. A stricter verifier can preserve declared distinctions;
it cannot determine ancient meaning or force a human to believe a caption.
The derived results locate two distinct losses: erasing conjunctions can
create unbounded apparent endorsement capacity, and deleting a shared guard can
unboundedly inflate the certified-reuse threshold. Retaining occurrence,
polarity and full-record frontier structure permits explicit feasibility and
revision checks under the declared policy.

\paragraph{Artifact and AI assistance.}
The \texttt{anc/} directory provides original licensed photographs, photogrammetry, losslessly wrapped artifact
models and unchanged published workbooks, source cards, protocol/optimizer code
and tests. Its offline runner regenerates raw results, plots and the image-scene,
envelope and linked early-evidence workbenches. It has no runtime
network requirement, trained weights or embedded PDF JavaScript. Generative AI
assisted literature discovery, mathematical exploration, drafting, implementation
and testing during preparation. This assistance is separate from the proposed
non-learned experimental method. The named human authors are responsible for
reviewing the arguments, sources, code and final submission; no independent
cryptographic audit or peer review is claimed. Data and code are packaged as arXiv ancillary files
with TeX sources~\cite{arxivanc}.

\bibliographystyle{plainnat}
\bibliography{references}
\appendix
\section{Reproduction and exact witness certificates}
From \texttt{anc/}, use Python 3.11 or later in a virtual environment. To
reproduce all image, geometry, source-table, revision, capacity, attack and plotting stages, run:
\begin{lstlisting}
python3 -m pip install -r requirements-lock.txt
python3 tools/reproduce.py --output ../scene-full-new --plots
\end{lstlisting}
Choose a fresh output directory: the runner refuses to overwrite prior results
or an existing ledger. It first restores and verifies exact source members, then records every stage
with plotting, raw measurements, structured
reports and regenerated figures. Individual commands are documented in the README.
Its deterministic test keys are intentionally public, and its fixed grants expire
relative to the fixed test clock. They must never represent real institutional
authority. \texttt{tools/reproduce.py} coordinates an isolated full run; dependency
locks record the measured environment. The workbench can export a draft JSON
scene for structural/source checking, but obtaining real signatures requires
the relevant issuer's informed authorization.

\paragraph{Checking the fractional certificate.}
Given rational $y_s\geq0$, independently compute every bundle price, each minimum
$p_{o,C}(y)$, and $D_C(y)$. If $D_C(y)>0$, equation~\eqref{eq:bound} certifies
$K_C\leq\lfloor Q\cdot y/D_C(y)\rfloor$ using exact arithmetic. This does not
certify an integral allocation exists. In the Fano fixture $y_s=1/3$ gives bound
$7/3$; pairwise intersections provide the additional exact upper bound of one.
The artifact retains all generated inputs and integer allocations so solver
claims can be checked against the independently enumerated finite problem.

\paragraph{Completeness of disclosure.}
For each selected occurrence, the verifier opens every referenced root evidence
record, follows all support IDs in the opened authenticated adjacency lists and
recursively demands those records. An induction on DAG depth establishes that
every reachable record is opened. Opening their source IDs then covers precisely
the charged sources. This argument establishes closure of the admitted graph,
not completeness of the archaeological record or an absent-edge claim about
the world. A source-list summary lacks the authenticated connectives needed to
state that induction.

\paragraph{Inspecting earlier evidence and revisions.}
The source XML readers keep original sheet/cell addresses and numerical lexemes.
\path|scenecrypt.gobekli_data|, \path|scenecrypt.social_analysis| and
\path|scenecrypt.early_holocene_analysis| rebuild the new source summaries.
\path|scenecrypt.revision_experiments| independently rebuilds the 10,000,
100,000 and 1,000,000-record revision branches. A fresh-run atlas is generated
from those run outputs, with original datasets read from the immutable input
folders. The published record universe excludes the extra skull variant and
uses only the twelve OBJ sources in the scaling experiment. Source snapshots,
retrieval programs, licenses, unchanged-workbook hashes and decompressed-mesh
hashes permit byte-level rechecking independently of the paper's plots.

\paragraph{Regional source comparisons.}
The frozen regional release manifest verifies original exports, workbooks,
licensed source snapshots and extracts before analysis. New stages regenerate
Chogha Mish context/anatomical counts, the strict NISP cohort, exact scope
certificates for both typed registries, shared-depth rounding witnesses, the
independent partial-transfer audit, all thirteen original-OBJ codec comparisons
and the regional companion. The complete source ZIP is strictly below the
user-selected 50,000,000-byte cap. All result/display caches are regenerated in
the external output directory; no source download is required.
\path|REGIONAL_RESEARCH_TR.md| gives a Turkish research brief and explicit future
falsification conditions. The \path|scope_contrasts| API binds an externally
admitted question; its policy descriptions are not an automatic historical
classification language. Original-source extraction and the current root remain
independent verification inputs. The packaged public keys/seeds are experimental
fixtures, and no runtime source download or trained model is required.

\paragraph{Language and external transfer.}
The default runner also executes the SQLite contract study and the complete
bundled museum-language experiment with its source-integrity tests. The optional
external temporal audit was executed for this release but requires two
separately acquired, hash-pinned corpora: provide both
\path|--ancient-times-zip| and \path|--wikiwars-zip|. Their narrative XML/labels
are not redistributed while explicit original redistribution permission remains
unresolved. The external evaluator and aggregate receipts are included; the
core pipeline remains independently reproducible offline after installing its
software/runtime dependencies. Exact JPEG reconstruction additionally requires
the trusted POSIX libjxl0.11.2 tools.

\end{document}